\documentclass[nohyperref]{article}

\usepackage{microtype}
\usepackage{graphicx}
\usepackage{subfigure}
\usepackage{booktabs} % for professional tables

\usepackage{hyperref}

\usepackage[accepted]{icml2023}

\usepackage{amsmath}
\usepackage{amssymb}
\usepackage{mathtools}
\usepackage{amsthm}
\usepackage{mathrsfs}

\usepackage[capitalize,noabbrev]{cleveref}

\usepackage{booktabs}
\usepackage{multirow}
\usepackage{multicol}
\usepackage{tabularx}
\usepackage{pgfplots}
\usepackage{diagbox}
\usepackage{slashbox}

\DeclareMathOperator{\RU}{\mathsf{RotUp}}
\DeclareMathOperator{\PRU}{\mathsf{PRotUp}}
\DeclareMathOperator{\RL}{\mathsf{RotLeft}}
\DeclareMathOperator{\CS}{\mathsf{SumCols}}
\DeclareMathOperator{\RS}{\mathsf{SumRows}}
\DeclareMathOperator{\cplx}{\mathsf{cplx}}
\DeclareMathOperator{\T}{\intercal}
\DeclareMathOperator{\Softmax}{\mathsf{Softmax}}
\DeclareMathOperator{\CE}{\mathsf{CE}}
\DeclareMathOperator{\conj}{\mathsf{Conj}}
\DeclareMathOperator{\lrot}{\mathsf{Lrot}}
\DeclareMathOperator{\rrot}{\mathsf{Rrot}}

\newcounter{protocol}
\newenvironment{protocol}[1]
{
    \par\addvspace{\topsep}
    \noindent
    \tabularx{\linewidth}{@{} X @{}}
    \toprule
    \refstepcounter{protocol}\textbf{Protocol \theprotocol} #1 \\
    \hline
    }
    {   \\
    \bottomrule
    \endtabularx
    \par\addvspace{\topsep}
}

\newcommand{\pt}{\textsf{pt}}
\newcommand{\ct}{\textsf{ct}}
\newcommand{\train}{\textsf{train}}
\newcommand{\val}{\textsf{val}}

\newcommand\enc[1]{\langle {#1} \rangle}

\newcommand\mrow[2]{\multirow{#1}{*}{#2}}
\newcommand{\sm}{\mathsf{Softmax}}

\theoremstyle{plain}
\newtheorem{theorem}{Theorem}[section]
\newtheorem*{theorem*}{Theorem}
\newtheorem{proposition}[theorem]{Proposition}
\newtheorem{lemma}[theorem]{Lemma}

\theoremstyle{definition}

\theoremstyle{remark}

\icmltitlerunning{HETAL: Efficient Privacy-preserving Transfer Learning with Homomorphic Encryption}

\begin{document}

\twocolumn[
\icmltitle{HETAL: Efficient Privacy-preserving Transfer Learning with Homomorphic Encryption}



\icmlsetsymbol{prev}{$\dag$}

\begin{icmlauthorlist}
\icmlauthor{Seewoo Lee
% \thanks{\newstuff{This work was done when Seewoo Lee was at CryptoLab Inc}}
}{bkl,prev}
\icmlauthor{Garam Lee}{ctl}
\icmlauthor{Jung Woo Kim}{ctl}
\icmlauthor{Junbum Shin}{ctl}
\icmlauthor{Mun-Kyu Lee}{inh}
\end{icmlauthorlist}

\icmlaffiliation{bkl}{University of California, Berkeley, US
% Work done at CryptoLab Inc.
}
\icmlaffiliation{ctl}{CryptoLab Inc., Seoul, South Korea}
\icmlaffiliation{inh}{Inha University, Incheon, South Korea  
% \newstuff{This work was done when Seewoo Lee was at CryptoLab Inc}
}

\icmlcorrespondingauthor{Mun-Kyu Lee}{mklee@inha.ac.kr}

\icmlkeywords{Privacy-preserving Machine Learning, Homomorphic Encryption, Transfer Learning, Softmax, Matrix Multiplication}

\vskip 0.3in
]



\printAffiliationsAndNotice{\icmlPrevWork} % otherwise use the standard text.

\begin{abstract}
\renewcommand{\thefootnote}{\fnsymbol{footnote}}
Transfer learning is a \emph{de facto} standard method for efficiently training machine learning models for data-scarce problems by adding and fine-tuning new classification layers to a model pre-trained on large datasets.
Although numerous previous studies proposed to use homomorphic encryption 
to resolve the data privacy issue in transfer learning in the machine learning as a service setting,
most of them only focused on encrypted inference. In this study, we present \textbf{HETAL}, an efficient \textbf{H}omomorphic \textbf{E}ncryption based \textbf{T}r\textbf{a}nsfer \textbf{L}earning algorithm, that protects the client's privacy in training tasks 
by encrypting the client data using the CKKS homomorphic encryption scheme.
\textbf{HETAL} is the first practical scheme that strictly provides encrypted training,
adopting validation-based early stopping and achieving the accuracy of nonencrypted training.
We propose an efficient encrypted matrix multiplication algorithm, which is 1.8 to 323 times faster than prior methods, and a highly precise softmax approximation algorithm with increased coverage.
The experimental results for five well-known benchmark datasets show total training times of 567--3442 seconds, which is less than an hour.
\footnote{Our codes for the experiments are available at https://github.com/CryptoLabInc/HETAL.}
\end{abstract}

\section{Introduction}

Transfer learning (TL)~\cite{FederatedLearning} is a \emph{de facto} standard method used to enhance the model performance by adding and fine-tuning new client-specific classification layers to a generic model pre-trained on large datasets.
In the machine learning as a service (MLaaS) setting,
a server may grant a client (data owner) access to a pre-trained model to extract features,
and the client can then send the extracted features to the server for fine-tuning.
Here, however,
sensitive client data can be leaked to the server because the extracted features may contain significant information about the original raw data.
For example, it is well-known that 
a facial image can be reconstructed from its feature vector~\cite{FaceReconstruction}.
There exist other recent studies demonstrating that feature reversion attacks pose severe threats against neural network-based face image recognition. For example, even SOTA face recognition systems, such as ArcFace and ElasticFace, are vulnerable to reversion attacks~\cite{Shahreza2022}. A recent study achieved a successful attack rate of 99.33\% against ArcFace features~\cite{Dong2023}.
In natural language processing (NLP), BERT embeddings can be inverted to recover up to 50--70\% of the original input words because of their semantic richness~\cite{song2020information}.
Therefore, 
a method to protect the transmitted features
is critical to protect the client's private data. 
Data privacy has become a worldwide concern~\cite{walch2022cryptotl},
with many countries having enacted privacy laws,
such as the EU General Data Protection Regulation
(GDPR) \cite{gdpr}.

To address this data privacy issue, extensive research has been conducted on privacy-preserving 
machine learning, % (PPML) and AI,
some of which can
be applied to TL. 
Most studies used cryptographic primitives such as secure multi-party computation (SMPC)~\cite{Kan2004MPC, Wan2007MPC, Nikolaeko2013MPC,  sameer2018MPC, 9076003}, differential privacy (DP)~\cite{Wang2019DP, Zhu2022DP}, and homomorphic encryption (HE)~\cite{walch2022cryptotl, van2019sealion, jin2020secure}.
Some previous studies have combined SMPC and HE~\cite{Nikolaeko2013MPC, Mohassel2017MPC, Lehmkuhl2021MPC, Chandran2022Com, hao2022iron}.
However, SMPC-based solutions require significant communication between the client and server and
DP-based solutions can reduce the accuracy. Although HE-based approaches can address these issues, they require extensive computations. Therefore, it is crucial to optimize HE operations to achieve practical performance.

SecureML \cite{Mohassel2017MPC} was
the first efficient privacy-preserving protocol for neural network training.
It effectively combined SMPC and linear HE but required two noncolluding servers for secure computation.
Elsloo et al.~\cite{van2019sealion} proposed SEALion, a HE-based solution for TL with a pre-trained VAE.
CryptoTL \cite{walch2022cryptotl} also used HE to secure TL. 
However, 
in SEALion and CryptoTL,
training for fine-tuning was not performed on the ciphertexts.
To be precise,
the server owns a private pre-trained model and
the client sends encrypted queries to the server. The server then performs inference
on the encrypted input and produces encrypted output features,
which can be decrypted by the client.
Fine-tuning is performed on these decrypted features by the client.
Since fine-tuning must be conducted on the client side, the client is expected to possess certain level of knowledge in machine learning,
but this is not always the case.
PrivGD~\cite{jin2020secure} is
the first HE-based MLaaS solution that supports encrypted fine-tuning of TL.
In this scenario, the client extracts features from its data using a shared feature extractor, encrypts the features using HE, and sends the encrypted features to the server to fine-tune the classifier.
However, PrivGD is designed for 
a small-scale sensor dataset with
an input feature dimension of 32.
The matrix multiplication algorithm in \cite{jin2020secure} could not be implemented using a typical GPU 
for a dataset with many features
owing to their high memory requirements.
For example, it requires more than 100GB for MNIST.

In this paper,  we aim to protect the client's privacy in training tasks for TL
when the client does not have the expertise in actual fine-tuning.
Therefore, we consider the same scenario as that of PrivGD~\cite{jin2020secure}.
We assume that the server is honest-but-curious (HBC). In other words, it does not deviate from the defined protocol but will attempt to learn all possible information from legitimately received messages~\cite{goldreich2004foundations}.
Although the server
fine-tunes the classifier,
it does not obtain the final model in plaintext
because the model is encrypted
with the client's key.
The server's training expertise is protected against a client as all training tasks are performed on the server side.

We propose
\textbf{HETAL}, an efficient \textbf{H}omomorphic \textbf{E}ncryption based \textbf{T}r\textbf{a}nsfer \textbf{L}earning algorithm for privacy-preserving TL.
\textbf{HETAL} is the first practical scheme that strictly provides HE-based encrypted training.
We applied the optimization techniques
used in non-encrypted training
and achieved almost the same accuracy as nonencrypted training for five well-known benchmark datasets.
We adopted validation-based early stopping, the most commonly used regularization method in deep learning~\cite{DLbook} to determine when to terminate the training.
To the best of our knowledge, none of the previous HE-based training methods have applied this because of performance issues.
For example,
PrivGD~\cite{jin2020secure}
fixed the number of training epochs
before starting the fine-tuning process,
considering the balance between the estimated multiplication depth in HE and accuracy,
where the balance was experimentally found in advance.
Our proposal for \textbf{HETAL} was achieved based on the significant acceleration of encrypted matrix multiplication, which is a dominant operation in the training task, and a highly precise approximation algorithm for the softmax function.
Our key contributions are as follows:
\begin{itemize}
\item We propose \textbf{HETAL}. 
\textbf{HETAL} protects the client's privacy in training tasks for TL 
by encrypting the client data using HE before sending it to the server. \textbf{HETAL} utilizes
the CKKS scheme~\cite{cheon2017homomorphic} because CKKS supports encrypted arithmetic over real numbers.
\item 
We implemented and evaluated \textbf{HETAL} using five well-known benchmark datasets (MNIST \cite{deng2012mnist},  CIFAR-10 \cite{krizhevsky2009learning}, Face Mask Detection \cite{larxel2020face},
DermaMNIST \cite{yang2023medmnist}, and SNIPS \cite{coucke2018snips}),
in addition to two pre-trained models (ViT \cite{dosovitskiy2021an} and MPNet \cite{song2020mpnet}). 
Our experimental results showed training times of 4.29 to 15.72 seconds per iteration and total training times of 567 to 3442 seconds (less than an hour), with almost the same classification accuracy as nonencrypted training. The accuracy loss by encrypted training was 0.5\% at most.

\item 
For \textbf{HETAL}, we propose a new softmax approximation algorithm, which covers a significantly wider range than the previous works with high precision.
We substantially expand the domain of softmax approximation to $[-128, 128]$, enabling us to train models for several hundred steps, which was impossible with
reasonable approximation error using previous approximation methods. For example, PrivGD \cite{jin2020secure} could not cover $[-8, 8]$. 
It was also impossible with direct application of the domain extension technique proposed in \cite{cheon2022efficient}.
We also provide a rigorous proof for the error bound of the proposed approximation algorithm.

\item
We also propose optimized matrix multiplication algorithms, $\mathsf{DiagABT}$ and $\mathsf{DiagATB}$, that compute matrix multiplications of the form $AB^{\T}$ and $A^{\T}B$ for encrypted matrices $A$ and $B$.
The outstanding speed of HETAL was aided by our optimization of matrix multiplication because it costs 18\% to 55\% of the overall training.
Our proposed algorithms are 
more efficient in both memory and computation than previous algorithms~\cite{jin2020secure,crockett2020low} and show a performance improvement of 1.8 to 323 times.

\end{itemize}

\section{Preliminaries}

\subsection{Transfer learning} \label{section:tl}

In this study, we focus on a multiclass classification task.
We adopted the most common approach for TL, using a pre-trained model as a feature extractor and fine-tuning a classification layer.
For fine-tuning, the layer is trained using Nesterov's accelerated gradient (NAG, \cite{nesterov1983method}) method to minimize the cross-entropy loss $\mathcal{L}_{\mathsf{CE}}$, which guarantees faster convergence than the vanilla SGD and is also HE-friendly (see \cite{kim2018logistic, kim2018secure, crockett2020low}).
More specifically, let $N$, $f$, and $c$ denote the mini-batch size, number of features, and number of classes, respectively. The number of features equals the output dimension of the pre-trained feature extractor.
Let $\mathbf{X} = (x_{ij}) \in \mathbb{R}^{N \times (f+1)}$ and $\mathbf{Y} = (y_{ik}) \in \mathbb{R}^{N \times c}$ be matrices representing the input features and one-hot encoded labels of the data in the mini-batch, respectively.
Let $\mathbf{W} = (w_{kj}) \in \mathbb{R}^{c \times (f +1)}$ be the parameter matrix.
We assume that the last column of $\mathbf{W}$ is the bias column and that the corresponding column of $\mathbf{X}$ is filled with ones.
The probability that the $i$-th data belongs to the $k$-th class is modeled using the softmax function as follows:
\begin{equation*}
\label{eqn:inference}
    \mathrm{prob}(\mathbf{X};\mathbf{W})_{ik} = \Softmax(\mathbf{X}_{i}\mathbf{W}^{\T})_{k},
\end{equation*}
where $\mathbf{X}_{i}\in \mathbb{R}^{f+1}$ denotes the $i$th row of $\mathbf{X}$.
We denote $\mathbf{P} = (\mathrm{prob}(\mathbf{X};\mathbf{W})_{ik}) \in \mathbb{R}^{N \times c}$.
Subsequently, the gradient $\nabla_{\mathbf{W}} \mathcal{L}_{\CE}$ of Cross-Entropy Loss $\mathcal{L}_{\mathsf{CE}}$ with respect to $\mathbf{W}$ has the simple form of
\begin{equation*}
\label{eqn:gradient}    
\nabla_{\mathbf{W}} \mathcal{L}_{\CE} = \frac{1}{N} (\mathbf{P} -\mathbf{Y})^{\T}\mathbf{X}.
\end{equation*}
We use it to update the layer's parameter $\mathbf{W}$ with NAG: For each step $t$,
\begin{align*}
    \mathbf{W}_{t+1} &= \mathbf{V}_{t} - \alpha \nabla_{\mathbf{V}_{t}} \mathcal{L}_{\CE}, \\
    \mathbf{V}_{t+1} &= (1 - \gamma_t) \mathbf{W}_{t+1} + \gamma_t \mathbf{W}_{t},
\end{align*}
where $\mathbf{V}_t$ are auxiliary parameters with randomly initialized $\mathbf{V}_1 = \mathbf{W}_1$, 
$\alpha$ denotes the learning rate,
and $\gamma_t = (1 - \lambda_{t}) / \lambda_{t+1}$ where $\lambda_0 = 0$ and $\lambda_{t+1} = (1 + \sqrt{1 + 4\lambda_{t}^2})/2$.

\begin{figure*}
  \begin{center}
  \includegraphics[width=0.5\textwidth]{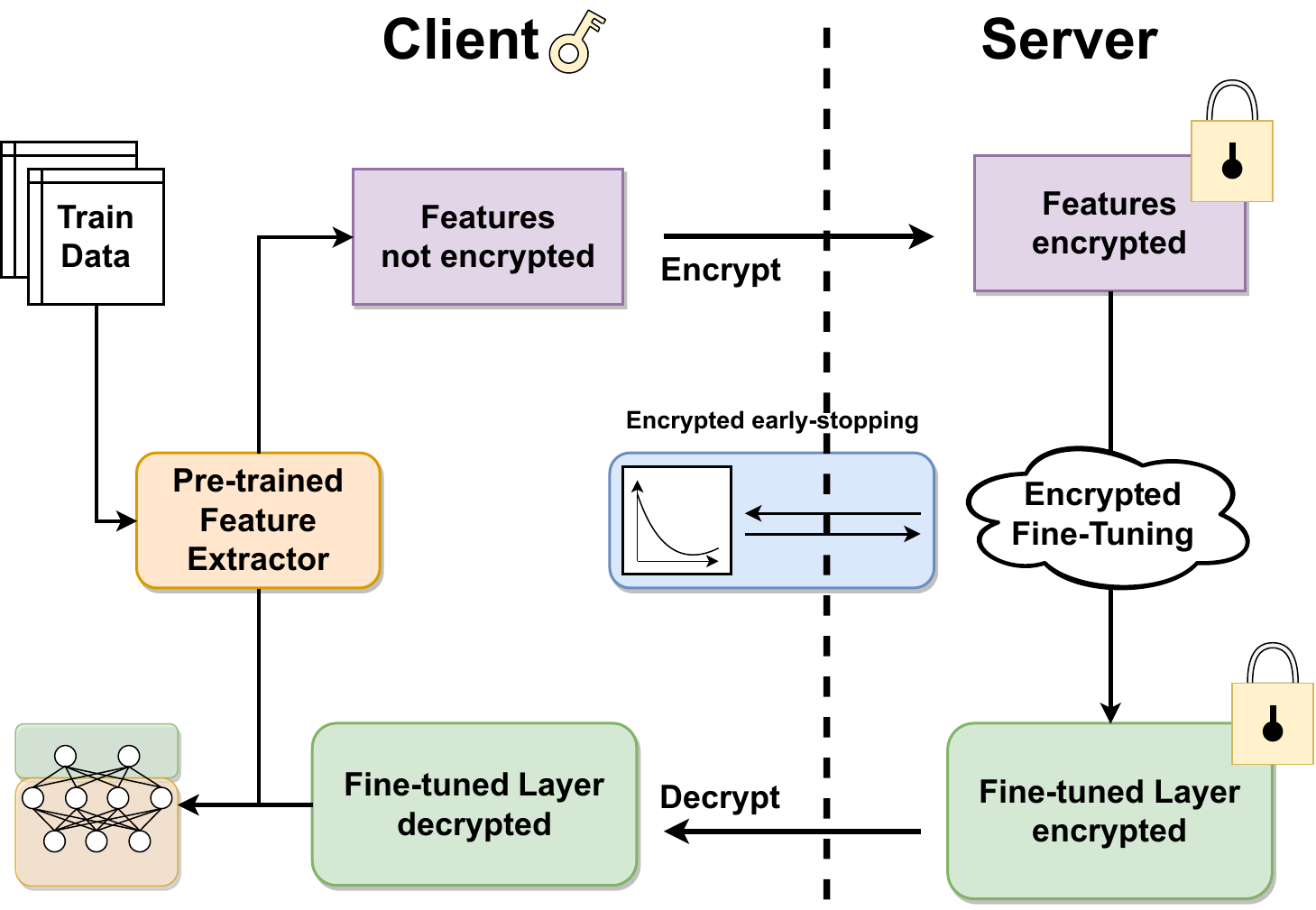}
  \end{center}
  \caption{Our privacy-preserving transfer learning protocol (\textbf{HETAL})}
  \label{fig:model}
\end{figure*}

\subsection{Homomorphic Encryption: CKKS Scheme}
Homomorphic Encryption (HE) is a cryptographic primitive that can support computations on encrypted data without decryption.
Particularly, the CKKS scheme \cite{cheon2017homomorphic} is an HE scheme, supporting \emph{approximate} arithmetic operations over encrypted real and complex numbers.
It encrypts multiple complex numbers into a single ciphertext and supports single instruction multiple data (SIMD) operations that perform the same operation simultaneously.
The following operations are available for ciphertexts: 
\begin{itemize}
    \item \textbf{Addition}: Element-wise addition of two ciphertexts.
    
    \item \textbf{Multiplication}: Element-wise multiplication of two ciphertexts.
    We denote $\mathsf{Mult}$ for ciphertext-ciphertext multiplication and $\mathsf{CMult}$ for plaintext-ciphertext multiplication (constant multiplication).
    We also denote $x \odot y$ for the multiplication of $x$ and $y$, irrespective of whether $x$ and $y$ are encrypted.
    The CKKS scheme is a \emph{leveled} HE scheme; therefore, we can compute multivariate polynomials of \emph{bounded} multiplicative depth.
    It is worth noting that multiplying an arbitrary complex constant also consumes a depth.
    
    \item \textbf{Rotation}: For a given plaintext $m = (z_{0}, \dots, z_{s-1}) \in \mathbb{C}^{s}$ and its ciphertext $\mathsf{ct}$, we can compute $\mathsf{ct_{rot}} = \mathsf{Lrot}(\mathsf{ct}, r)$ for $0\leq r < s$, whose decryption is approximately $(z_{r}, \dots, z_{s-1}, z_{0}, \dots, z_{r-1})$.

    \item \textbf{Complex conjugation}: Element-wise complex conjugation of a ciphertext.
    
    \item \textbf{Bootstrapping}: 
    Bootstrapping \cite{cheon2018bootstrapping} is a unique operation that allows us to compute multivariate polynomials of \emph{arbitrary} degrees.
    Since bootstrapping is the most expensive computation among all the basic operations in homomorphic encryption, it is crucial to reduce the multiplicative depths of circuits to reduce the number of bootstrapping operations.
\end{itemize}

% \vspace{-1ex}
\subsection{Threat Model}

We assume an AutoML-like service, in which a client can outsource model training to a server. The proposed system comprises two parties: one is a client who owns the data and the other is a server that provides ML training services. The protocol aims to allow clients to outsource model training to a server while preserving the privacy of their data. We assume that the server and client can share a pre-trained generic model as a feature extractor, because there are many publicly available pre-trained models, including Vision Transformer (ViT) used in our experiments. During the training task, the client extracts features from its private data using the feature extractor and sends the HE-encrypted features to the server. The server performs fine-tuning for TL on the ciphertext domain and produces an encrypted model.

We assume that the server is honest-but-curious (HBC), where an HBC adversary is a legitimate participant in a communication protocol who will not deviate from the defined protocol but will attempt to learn all possible information from legitimately received messages~\cite{goldreich2004foundations}. Note that HE provides a good defense to protect the data against an HBC server because the server performs computation over encrypted data without knowing the decryption key.
\section{Protocol}

% \vspace{-2ex}

\subsection{HETAL Protocol}

Protocol 1 describes \textbf{HETAL}, which is our privacy-preserving transfer learning protocol.
We use superscript $\ct$ (resp. $\pt$) when the matrix is encrypted (resp. not encrypted).
Note that the client's data are encrypted using the client's key, so the server cannot decrypt it, while the server can perform computations using public operation keys.
The client extracts features from the input data using a pre-trained model (step 1) and sends the encrypted features to the server (step 2).
Then the server fine-tunes a classification layer with encrypted data using NAG, which results in an encrypted fine-tuned layer (step 3).
The server sends the encrypted final model to the client and lets the client decrypt and use it for inference (step 4).
The server can early-stop the training by computing logits of a validation set and communicating with the client (step 3 (b)-(d)).
The client computes the validation loss with the (decrypted) logit and labels and then sends a signal to stop the training if needed, which can prevent overfitting.
The client's data are not revealed to the server because they remain encrypted during the training. 
Figure~\ref{fig:model} shows the overall procedure of \textbf{HETAL}.

We remark that as an alternative to step 4 of our protocol, it is also possible to store the encrypted layer in the server for encrypted inference: the client may send the encrypted features to be classified to the server, receive the encrypted result
and decrypt it.

\begin{protocol}{Protocol of \textbf{HETAL}}
\vspace{-1.5ex}
\begin{enumerate}
    \item Using a pre-trained model, the client extracts features from training and validation data.
    \item The client encrypts the extracted features and labels of the training set as $\mathcal{B}_{\train}^{\ct} = \{(\mathbf{X}_{\train}^{\ct}, \mathbf{Y}_{\train}^{\ct})\}$ and sends them to the server. The client also sends the encrypted features $\mathbf{X}_{\val}^{\ct}$ of the validation set to the server.
    \item For each epoch, repeat the following:
    \begin{enumerate}
        \item For $(\mathbf{X}^{\ct}, \mathbf{Y}^{\ct}) \leftarrow \mathcal{B}^{\ct}_{\train}$, server updates parameters with NAG:
        $$
            \mathbf{W}^{\ct}_{t+1} = \mathbf{V}^{\ct}_{t} - \frac{\alpha}{N} (\mathbf{P}^{\ct} - \mathbf{Y}^{\ct})^{\T} \mathbf{X}^{\ct}
        $$
        $$
            \mathbf{V}_{t+1}^{\ct} = (1 - \gamma_t)\mathbf{W}_{t+1}^{\ct} + \gamma_t \mathbf{W}_{t}^{\ct}
        $$
        where $\alpha$ is the learning rate, $N$ is the batch size, 
        $\mathbf{P}^{\ct} = \mathsf{ASoftmax}(\mathbf{X}^{\ct}(\mathbf{V}^{\ct}_{t})^{\T})$ with (row-wise) softmax approximation $\mathsf{ASoftmax}$, 
        and $\{\gamma_t\}_{t\geq 0}$ is defined in Section \ref{section:tl}.
        \item With the last weight $\mathbf{W}^{\ct}$ in the epoch, the server computes logits $\ell_{\val}^{\ct} = \mathbf{X}^{\ct}_{\val}(\mathbf{W}^{\ct})^{\T}$ and sends them to the client.
        \item The client decrypts the logits and computes the validation loss $\mathcal{L}_{\CE}$ with $\ell_{\val}^{\pt}$ and $\mathbf{Y}_{\val}^{\pt}$.
        \item If a loss is not decreased for a fixed number of epochs (patience), then the client sends a signal to the server to stop training.
        Otherwise, the current best weight $\mathbf{W}^{\ct}_{*}$ is replaced with the new weight $\mathbf{W}^{\ct}$.
    \end{enumerate}
    \item The server sends $\mathbf{W}_{*}^{\ct}$ to the client, and the client decrypts the parameter and obtains the fine-tuned layer $\mathbf{W}^{\pt}_{*}$.
\end{enumerate}
\end{protocol}

% \vspace{-4ex}

\subsection{Security Analysis}

We will demonstrate that HETAL protects the client’s data against an HBC server. If the number of training epochs is fixed, the only information a server can see is encrypted features. Therefore, HETAL protects client’s data against HBC server owing to the security of HE. Even when a validation-based early stopping (Step 2 (b) – (d)) is used, only one additional information, i.e., a signal to let the server stop training, is sent to the server, which does not seem to be useful for the server to recover the client’s private data. Additionally, HETAL protects the fine-tuned model from the server because only a client can see it.

In practice, when HETAL is used in AutoML-like MLaaS, the client software may be provided by the server because MLaaS should be available to a non-expert in ML. The client can simply run the software, which applies a pre-trained model to its own data, and then encrypt the extracted features using HE before sending them. However, in this scenario, the client’s data can be secured only if the software is trustworthy. If the source code of the client software can be open to the public, we can apply an approach such as code verification by third parties: for example, see Section 3.1 of \cite{knauth2018integrating}, stating, ``interested parties can inspect the source code to convince themselves.''
Therefore, we can ensure the trustworthiness of the HETAL client SW by making them publicly verifiable: it is not harmful for a server to open the source code to public because it consists of a known pre-trained model and encryption function, which are not the server's secrets.

Finally, we mention that HETAL protects the server's expertise, such as hyper-parameters and optimized software for training because all training tasks are performed on the server side.

\section{Algorithm}

In this section,
we propose a new softmax approximation algorithm with a much wider range than those in previous studies.
In addition, we propose two novel encrypted matrix multiplication algorithms, denoted $\mathsf{DiagABT}$ and $\mathsf{DiagATB}$, which compute matrix multiplications of the forms $AB^{\T}$ and $A^{\T}B$.

\subsection{Softmax Approximation}

There are several works on the approximation of softmax function with polynomials \cite{lee2022privacy, hong2022secure, jin2020secure}.
However, all of these methods have low precision (See the Table \ref{tab:softmaxerr} in the Appendix) and permit only a small domain of approximation, which is not desirable for training with many epochs.
We remark that both \cite{lee2022privacy} and \cite{hong2022secure} have been used for inference rather than training.

Regarding the domain of approximation,
we experimentally found that the input values of the softmax function increase as training proceeds and easily deviate from the domain that the previous methods can handle.
For example, the maximum input value of softmax varies from 0.38 to over 100 on MNIST dataset  (Figure \ref{fig:softmaxinput}), which cannot be controlled by the previous approximation methods.
Consequently, it is essential to expand the domain of approximation to perform as many training epochs as we want.

To approximate the softmax function on a large interval efficiently, we apply the domain extension technique by \cite{cheon2022efficient}.
The authors introduced domain extension functions (DEFs) and domain extension polynomials (DEPs) and provided an algorithm to approximate a sigmoid-like functions on exponentially large intervals.
More precisely, they provided an algorithm to approximate a DEF that clips an input into a fixed interval as a composition of low-degree polynomials, and applied it to obtain a polynomial that approximates the sigmoid function $\sigma(x) = 1 / (1 + e^{-x})$ on a large interval of scale $[-7683, 7683]$.
In general, we can approximate a polynomial on a range $[-R, R]$ with $O(\log R)$ complexity.

However, directly applying the domain extension algorithm in \cite{cheon2022efficient} does not constantly work.
For example, assume that we have an approximation of a 3-variable softmax on a 3-dimensional box $[-8, 8]^{3}$, and an input is given by $(8, 10, 13)$.
Here, the actual value of softmax is $\mathsf{Softmax}(8, 10, 13) = (0.006, 0.047, 0.946)$.
However, the naive application of the DEF clips the input as $(8, 8, 8)$ and produces $(0.333, 0.333, 0.333)$ as an output.
A naive application of the method can lead to identical outputs and result in considerable errors.
To address the issue, we first \emph{normalize} input by subtracting maximum value.
More precisely, we first compute approximate maximum $m = \mathsf{Amax}(\mathbf{x})$ and set $\mathbf{x}' = \mathsf{Norm}(\mathbf{x}) = (x_1', \dots, x_j') \in \mathbb{R}^c$ as $x'_j = x_j - m$ for $1\leq j \leq c$.
We use the homomorphic comparison algorithm proposed in \cite{cheon2020efficient} to compute $\mathsf{Amax}$
homomorphically as $\mathsf{Amax}(a, b) = a \cdot \mathsf{Acomp}(a, b) + b \cdot (1 - \mathsf{Acomp}(a, b))$, with $O(\log c)$ comparisons and rotations,
where $\mathsf{Acomp}$ approximates a function $\mathsf{comp}(a, b)$ that returns 1 if $a \geq b$ and 0 otherwise.
Then we can easily see that $\sm(\mathbf{x}') = \sm(\mathbf{x})$.
Now let $D_n: [-L^n R, L^n R]^c \to [-R,  R]^c$ be a DEP obtained by Algorithm 1 of \cite{cheon2022efficient} with $n$-iterations, where $L$ is a domain extension ratio.
The following theorem tells us that, if $p$ is an approximation of the softmax on a small domain, then $\mathsf{ASoftmax}:=p \circ D_n \circ \mathsf{Norm}$ gives an approximation of softmax on a large domain.
For example, when $R=8$, $L=2$, and $n =5$, $\mathsf{ASoftmax}$ may cover $[-128, 128]$, which can handle all steps in Figure~\ref{fig:softmaxinput}.

\begin{theorem*}%[informal]
Let $p: \mathbb{R}^{c} \to \mathbb{R}^{c}$ be an approximation of the softmax on $[-R, R]^{c}$ satisfying 
$$
\|\mathsf{Softmax}(\mathbf{x}) - p(\mathbf{x})\|_{\infty} < \epsilon.
$$
Then for $\mathbf{x} \in [-\frac{1}{2}L^n R, \frac{1}{2}L^n R]^c$, we have
$$
\|\mathsf{Softmax}(\mathbf{x}) - p(D_n(\mathsf{Norm}(\mathbf{x})))\|_{\infty} < \beta + \epsilon,
$$
where $\beta = \beta (\delta, c, r, L, d)$ is a constant that depends only on $\delta, c, r,L, d$.

\end{theorem*}
The proof is stated in Appendix~\ref{appendix:softmax}
with a formal version of the theorem. (See Theorem \ref{thm:softmaxbound}.)
The theorem implies that
we can approximate softmax on a given domain $[-R_0, R_0]^c$ with $O(\log R_0)$ operations.
The precise algorithm can be found in Algorithm \ref{alg:softmax} in the Appendix.

\subsection{Encrypted Matrix Multiplication}
We first explain how to encode a matrix into ciphertext(s) and how to perform encrypted matrix multiplications of the forms $AB^{\T}$ and $A^{\T}B$ to compute logits and gradients in \textbf{HETAL}.
The main goal of our algorithm is to reduce the number of rotations and multiplications used, which occupy most of the algorithm runtime (see Section \ref{sec:5.1} for the costs of each operation in HE).
Note that including a transpose in multiplication is more efficient than computing matrix multiplication of the form $AB$ directly as in \cite{jiang2018secure, huang2021more}, since we need to perform transpose for each iteration of training, which is a costly operation and requires additional multiplicative depth.

\begin{figure}
    \centering
    \includegraphics[width=0.45\textwidth]{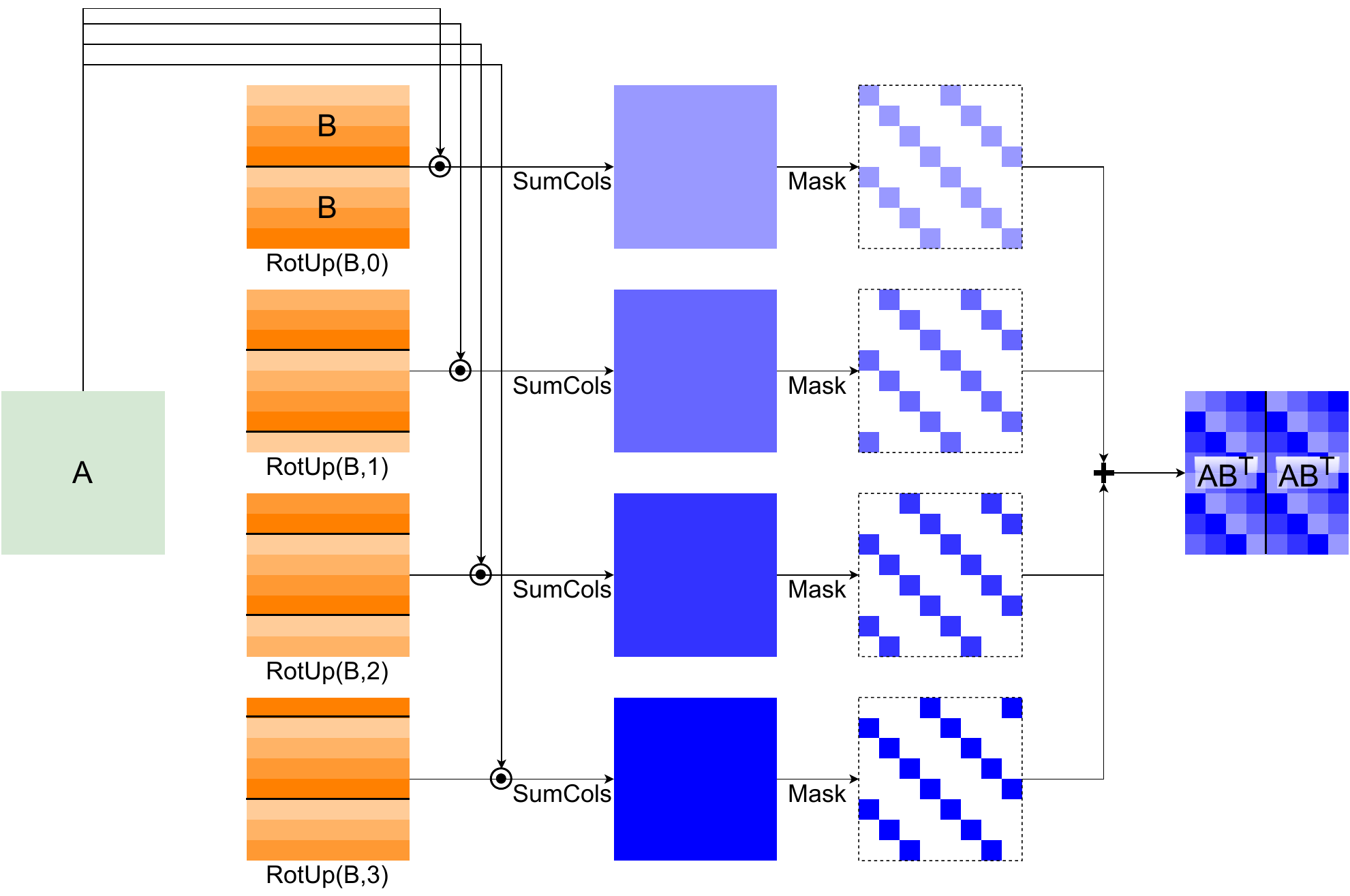}
    \caption{Demonstration of \textsf{DiagABT} algorithm when $A \in \mathbb{R}^{8 \times 8}$ and $B \in \mathbb{R}^{4 \times 8}$. We do not include the complexification optimization in the figure for simplicity.}
    \label{fig:diagabt}
\end{figure}

\subsubsection{Encoding}
We used the same encoding method as in \cite{crockett2020low},
dividing a given matrix as submatrices of a fixed shape $s_0 \times s_1$ and encode each submatrix in the row-major order.
Here, each submatrix corresponds to a single ciphertext so that the number of entries in each submatrix equals the number of slots in a single ciphertext, i.e., $s_{0}s_{1} = s$. 
For convenience, we assume that all matrices are sufficiently small to fit into a single ciphertext.
We also assume that the number of rows and columns of the matrices are powers of two by applying zero padding when required.
The algorithms can be easily extended to larger matrices (composed of multiple ciphertexts). The details can be found in the Appendix.

\subsubsection{Computation of $AB^{\T}$}
Let $A \in \mathbb{R}^{a \times b}$ and $B \in \mathbb{R}^{c \times b}$ be two matrices. 
Our goal is to compute $AB^{\T} \in \mathbb{R}^{a \times c}$ using basic HE operations such as addition, multiplication, and rotation.
For this, we use tiling, off-diagonal masking and complexification to reduce the computational complexity.

First, we define some operations and notations.
For a given matrix $B$, we define $\RU(B, k)$ as a matrix $B'$ obtained by rotating the rows of $B$ in the upper direction by $k$, i.e.,
$$
B'_{i, j} = B_{(i + k) \,\mathrm{mod}\, c, j}.
$$
When $B$ is encoded in a row-wise manner, $\RU(B, k)$ can be obtained from $B$ by simply applying the left rotation of index $kb$, i.e., $\RU(B, k) = \lrot(B, kb)$.

Next, for an $s_{0} \times s_{1}$ matrix $X$,  we define $\CS(X)$ as a matrix with entries
$$
\CS(X)_{i, j} = \sum_{0\leq k < s_{1}}X_{i, k}.
$$
In other words, each column of $\CS(X)$ is the sum of the columns in $X$.
This can be computed with $2\log s_{1}$ rotations and one constant multiplication: see Algorithm 2 in \cite{crockett2020low} for more detail.

We also define matrix $\overline{B}$ as an $s_{0} \times b$ matrix, where $(s_{0}/c)$ copies of $B$ are tiled in the vertical direction.
Finally, we define $\overline{B}_{\cplx}$, \emph{complexification of $\overline{B}$}, as
$$
\overline{B}_{\cplx} = \overline{B} + \sqrt{-1} \RU(\overline{B}, c/2).
$$
Note that multiplying $i = \sqrt{-1}$ does not consume a multiplicative depth.

Using the operations defined above, we compute $A\overline{B}^{\T}$ as follows.
Note that $A\overline{B}^{\T}$ is a matrix containing $(s_{1} / c)$ copies of $AB^{\T}$ in the horizontal direction.
\begin{proposition}
\label{prop:abt}
Let $A$ and $B$ be defined as above.
We have $A\overline{B}^{\T} = X + \conj(X)$, where
% \small
\begin{equation*}
X = \sum_{0 \leq k < c/2} \CS(A \odot \RU(\overline{B}_{\cplx}, k)) \odot M_{\cplx}^{(k, c)}.
\label{eqn:cplxabt}
\end{equation*}
Here $M^{(k, c)}$ is an off-diagonal masking matrix with entries
$$
M^{(k, c)}_{i, j} = 
    \begin{cases}
        1 & j \equiv i + k \,(\mathrm{mod}\, c) \\
        0 & \text{otherwise}
    \end{cases}
$$
and $M_{\cplx}^{(k, c)}$ is a \emph{complexified} version of the mask given by
$$
M_{\cplx}^{(k, c)} = \frac{1}{2} M^{(k, c)} - \frac{\sqrt{-1}}{2} M^{(k + c/2, c)}.
$$
\end{proposition}

Figure~\ref{fig:diagabt} illustrates the proposed multiplication method based on Proposition \ref{prop:abt}.
The detailed procedure for the proposed algorithm, $\mathsf{DiagABT}$,
is presented as
Algorithm~\ref{alg:abt} in the Appendix.
Note that the number of rotations is a bottleneck for matrix multiplication, and tiling reduces it from $O(s_{0}\log s_{1})$ to $O(c\log s_{1})$.
This also fits into our case for computing $\mathbf{X}\mathbf{W}^{\T}$, because the number of rows of $\mathbf{W}$ equals the number of classes for the dataset, which is often small compared to $s_{0}$ or $s_{1}$.
In addition, complexification reduces the complexity by half
(a similar idea was used in \cite{hong2022secure}).
Finally, we can extend the algorithm to compute $tAB^{\T}$ for $t \in \mathbb{R}$ without additional multiplicative depth consumption by replacing the diagonal mask $M^{(k, c)}_{\cplx}$ with $tM^{(k, c)}_{\cplx}$.
Algorithm~\ref{alg:abt} adopts this optimization.

\subsubsection{Computation of $A^{\T}B$}
We use similar ideas as $AB^{\T}$ for efficient computation of $A^{\T}B$. In addition, we propose a new operation, i.e., partial rotation, to reduce multiplicative depth.
Let $A\in \mathbb{R}^{a \times c}$ and $B \in \mathbb{R}^{a \times b}$.
We can similarly compute $A^{\T}B$.
First, we define $\RL({A}, k)$ as a matrix $A'$ obtained by rotating the columns of $A$ in the left direction by $k$, i.e., 
$$
A'_{i, j} = A_{i, (j + k) \mathrm{mod\,} c}.
$$
Unlike $\RU$, this consumes a multiplication depth.

Similar to $\CS$, we can also define $\RS(X)$ for an $s_{0} \times s_{1}$ matrix $X$ as
$$
\RS({X})_{i, j} = \sum_{0 \leq k < s_{0}} X_{k, j}.
$$
In other words, each row of $\RS({X})$ is the sum of the rows in $X$.
This can be achieved using with $\log s_{0}$ rotations without additional depth consumption.

As in the case of $AB^{\T}$, we can apply tiling and complexification to reduce the computational complexity.
We denote $\underline{A}$ for an $a \times s_{1}$ matrix where $(s_{1}/c)$ copies of $A$ are tiled in the horizontal direction.
We also define complexification $\underline{A}_{\cplx}$ of $\underline{A}$ as
$$
\underline{A}_{\cplx} = \underline{A} + \sqrt{-1} \RL(\underline{A}, c/2).
$$
However, since $\RL$ consumes a multiplicative depth, the level of $\underline{A}_{\cplx}$ is smaller than that of $\underline{A}$ by one. This may increase the multiplicative depth of $\underline{A}^{\T}B$ by one when the level of $\underline{A}$ is smaller than that of $B$.
To address this issue, we propose an algorithm that eventually consumes $B$'s level instead of $A$'s.
We first define a new operation $\PRU(-, k)$.
For a matrix $B$, $\PRU(B, k)$ is defined as a matrix that the last $k$ columns are rotated upwards by one position.
For example, the following figure shows its effect on a matrix of shape $4 \times 8$ with $k = 3$.
% }

\vspace{-1ex}
\begin{figure}[h]
    \centering
    \includegraphics[width=0.35\textwidth]{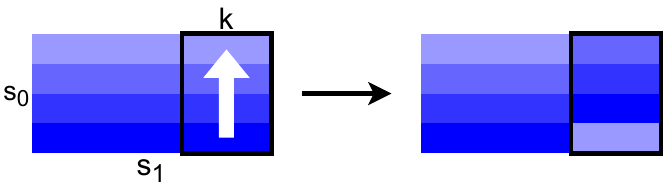}
    \label{fig:pru}
\end{figure}
\vspace{-1ex}
We can compute this homomorphically with a single $\mathsf{CMult}$ and $\lrot$, consuming a multiplicative depth (see Algorithm \ref{alg:pru} in Appendix).
Using this new operation $\PRU$, $\underline{A}^{\T}B$ can be expressed as follows:
\begin{proposition}
\label{prop:atbdepth}
$\underline{A}^{\T}B =  X + \conj(X)$, where 
{\scriptsize
\[
X = \sum_{0 \leq k < c/2} \RS(\lrot(\underline{A}_{\cplx}, k) \odot \PRU(B, k)) \odot M_{\cplx}^{(-k, a)}. 
\]
}
% \normalsize
\end{proposition}
Like $AB^{\T}$, tiling and complexification has an effect of reducing the number of rotations from $O(s_{1}\log s_{0})$ to $O(c\log s_{0})$.
Algorithm \ref{alg:atbcplxlowdepth} of Appendix adopts this optimization.

\section{Experimental results}\label{section:experiment}

\subsection{Experimental setup}
\label{sec:5.1}

We used HEaaN \cite{heaanlib}, a homomorphic encryption library based on the RNS version of the CKKS scheme \cite{cheon2018full}, which supports bootstrapping \cite{cheon2018bootstrapping} and GPU acceleration.
We take $2^{16}$ as a cyclotomic ring dimension (so that each ciphertext has $s = 2^{16-1} = 32768$ slots) and ciphertext modulus $q \approx 2^{1555}$, which ensures a 128-bit security level under the SparseLWE estimator \cite{cheon2022practical}.

We used an Intel Xeon Gold 6248 CPU at 2.50GHz, running with 64 threads, and a single Nvidia Ampere A40 GPU.
Each operation's execution time is measured as follows: \textsf{Add}: 0.085 ms, \textsf{Rotate}: 1.2 ms,  \textsf{CMult}: 0.9 ms, \textsf{Mult}: 1.6 ms, and \textsf{Bootstrap}: 159 ms.

\begin{table*}[t]
\centering
\setlength\extrarowheight{2pt}
\resizebox{2\columnwidth}{!}{%
\begin{tabular*}{\textwidth}{@{\extracolsep{\fill}}*{6}{c}}
\toprule
\multirow{3}{*}{dataset}& \multicolumn{3}{c}{encrypted} & \multicolumn{1}{c}{not encrypted} \\
\cline{2-4}  \cline{5-5}
& \multicolumn{2}{c}{Running time} & \multirow{2}{*}{ACC (a)} &  \multirow{2}{*}{ACC (b)} &\multirow{2}{*}{ACC loss ((b) - (a))}\\
\cline{2-3}
& Total (s) & Time / Iter (s) & & &  \\
\cline{1-6}
MNIST & 3442.29 & 9.46 & 96.73\% & 97.24\% & 0.51\%\\
CIFAR-10 & 3114.30 & 15.72 & 96.57\% & 96.62\% & 0.05\% \\
Face Mask Detection & 566.72 & 4.29 & 95.46\% & 95.46\% & 0.00\% \\
DermaMNIST & 1136.99 & 7.06 & 76.06\% & 76.01\% & -0.05\% \\
SNIPS & 1264.27 & 6.95 & 95.00\% & 94.43\% & -0.57\% \\
\bottomrule
\end{tabular*}
}
\caption{Transfer learning results on 5 benchmark datasets.}
\label{tab:transfer}
\end{table*}

\subsection{Transfer learning}

We used five benchmark datasets for image classification and sentiment analysis: MNIST \cite{deng2012mnist}, CIFAR-10 \cite{krizhevsky2009learning}, Face Mask Detection \cite{larxel2020face}, DermaMNIST \cite{yang2023medmnist}, and SNIPS \cite{coucke2018snips}.
These benchmarks are widely used or private.
In addition, We used the pre-trained ViT \cite{dosovitskiy2021an} (\texttt{ViT-Base}) and MPNet \cite{song2020mpnet} (\texttt{MPNet-Base}) as feature extractors for image and natural language data, respectively.
Both models embed a data point into a single 768-dimensional vector.

The results are shown in Table \ref{tab:transfer}.
For early stopping, we set the patience to 3.
Table \ref{tab:transfer} shows that we fine-tuned the encrypted models on all benchmark datasets within an hour.
In addition, the accuracy drops of the encrypted models were at most 0.51\%, compared to the unencrypted models with the same hyperparameters.
(See Table \ref{tab:hyperparam} in the Appendix for hyperparameters
and the number of epochs for early-stopping.).
The time required to transmit logits on the validation datasets is negligible compared to the total execution time,
because the total size of ciphertexts for encrypting logits is at most 8.8 MB for each epoch.

We also note that our method scales robustly with larger models, such as \texttt{ViT-Large} that embeds a datapoint into a 1024-dimensional vector (see Appendix \ref{appendix:largermodel} and Table \ref{tab:largevit}).

\vspace{-1ex}
\subsection{Softmax Approximation}

In the above experiments, we used our new softmax approximation that covers inputs in the range $[-128, 128]$
(See the Appendix for the detailed parameters).
To estimate the approximation error,
we used Monte Carlo simulation,
because it is computationally intractable to find the exact maximum error of functions with many variables.
To be precise, we sampled 300 M points on the domain and found that the maximum error was 0.0037--0.0224 and the average error was 0.0022--0.0046 depending on the input dimension.
Table~\ref{tab:softmaxerr} in Appendix \ref{appendix:softmaxcomparison} shows that these errors are significantly smaller than those of the previous methods.

% \vspace{-8ex}
Figure~\ref{fig:softmaxinput} explains the reason for this improvement.
It shows how the minimum and maximum values of input of softmax vary as the training proceeds.
The value increases in the order of two, which cannot be handled by previous approximation methods \cite{lee2022privacy, hong2022secure, jin2020secure}.
This shows that, with the previous approximation methods, it is hard to train a model as much as we want.

% \pgfplotsset{compat = 1.10}
\pgfplotstableread[col sep = comma]{plots/cifar10.csv}\cifardata
\pgfplotstableread[col sep = comma]{plots/mnist.csv}\mnistdata
\pgfplotstableread[col sep = comma]{plots/facial-mask-detection.csv}\maskdata
\pgfplotstableread[col sep = comma]{plots/dermamnist.csv}\dermadata
\pgfplotstableread[col sep = comma]{plots/snips.csv}\snipsdata

\begin{figure}
    \centering
    \begin{tikzpicture}
      \begin{axis}[
        legend pos = north west,
        legend style={nodes={scale=0.5, transform shape}},
        xmin = 0, xmax = 365,
        ymin = -64, ymax = 128,
        ytick = {-64, -32, -8,  8, 32, 64, 128}
        ]
        \addplot[line width = 1.0pt, red] table [x index = {0}, y index = {1}] {\mnistdata};
        \addplot[line width = 1.0pt, orange] table [x index = {0}, y index = {1}] {\cifardata};
        \addplot[line width = 1.0pt, green] table [x index = {0}, y index = {1}] {\maskdata};
        \addplot[line width = 1.0pt, blue] table [x index = {0}, y index = {1}] {\dermadata};
        \addplot[line width = 1.0pt, purple] table [x index = {0}, y index = {1}] {\snipsdata};
        \addplot[mark=none, gray, line width=0.1pt] coordinates {(0,32) (365,32)};
        \addplot[mark=none, gray, line width=0.1pt] coordinates {(0,-32) (365,-32)};
        \addplot[mark=none, gray, line width=0.1pt] coordinates {(0,8) (365,8)};
        \addplot[mark=none, gray, line width=0.1pt] coordinates {(0,-8) (365,-8)};
        \addplot[mark=none, gray, line width=0.1pt] coordinates {(0,4) (365,4)};
        \addplot[mark=none, gray, line width=0.1pt] coordinates {(0,-4) (365,-4)};
        \addplot[line width = 1.0pt, red] table [x index = {0}, y index = {2}] {\mnistdata};
        \addplot[line width = 1.0pt, orange] table [x index = {0}, y index = {2}] {\cifardata};
        \addplot[line width = 1.0pt, green] table [x index = {0}, y index = {2}] {\maskdata};
        \addplot[line width = 1.0pt, blue] table [x index = {0}, y index = {2}] {\dermadata};
        \addplot[line width = 1.0pt, purple] table [x index = {0}, y index = {2}] {\snipsdata};
        \legend{MNIST,CIFAR-10,Facial Mask Detection,DermaMNIST,SNIPS}
      \end{axis}
    \end{tikzpicture}
    \caption{Maximum and minimum value of input of softmax at each step (minibatch) for each dataset.}
    \vspace{-2ex}
    \label{fig:softmaxinput}
\end{figure}
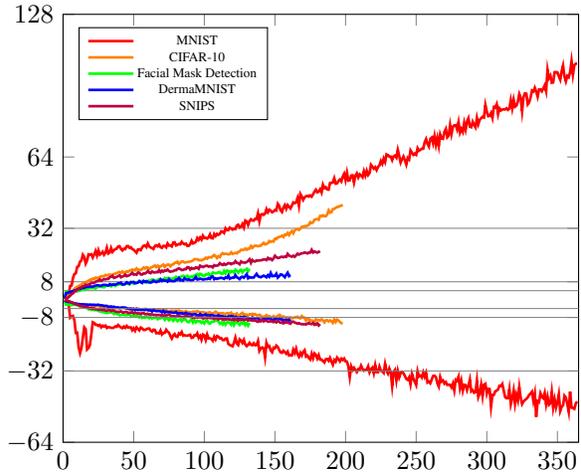

\begin{table*}[ht]
\centering
\resizebox{2\columnwidth}{!}{%
\setlength\extrarowheight{2pt}
\setlength\tabcolsep{1.5pt}
% \begin{tabular}{@{\extracolsep{5pt}}*{7}{c}}
% \begin{tabular}{@{\extracolsep{3pt}}*{9}{c}}
\begin{tabular}{c|ccc|c|c|ccc|c|c}
% \begin{tabular}{ccccc}
% {c|cccc|cc}
\toprule
\multirow{2}{*}{$(a, b, c)$}& \multicolumn{5}{c|}{$AB^{\T}$ ($A \in \mathbb{R}^{a\times b}, B \in \mathbb{R}^{c \times b}$)} & \multicolumn{5}{c}{$A^{\T}B$ ($A \in \mathbb{R}^{a \times c}, B \in \mathbb{R}^{a \times b}$)} \\
% \cline{2-5}  \cline{6-9}
\cline{2-11}
 & \cite{jin2020secure}$^{*}$ & \textsf{ColMajor}$^\dagger$ & \textsf{DiagABT} & \multicolumn{2}{c|}{Speedup} & \cite{jin2020secure}$^{*}$ & \textsf{RowMajor}$^\dagger$ & \textsf{DiagATB} & \multicolumn{2}{c}{Speedup} \\
\midrule
$(128, 128, 4)$ & 0.8192  & 0.1104 & \textbf{0.0601} & 13.63 & 1.84 & 10.0352 & 0.1171 & \textbf{0.0415} & 241.81 & 2.82 \\
$(256, 256, 8)$ & 3.2768 & 0.3203 & \textbf{0.1211} & 27.06 & 2.64 & 40.1408 & 0.3167 & \textbf{0.1239} & 323.98 & 2.56 \\
$(512, 769, 4)$ & 4.9216 & 0.7609 & \textbf{0.1223} & 40.24 & 6.22 & 60.2896 & 0.7176 & \textbf{0.3343} & 180.35 & 2.15 \\
$(1024, 769, 8)$ & 9.8432 & 3.0428 & \textbf{0.3710} & 26.53 & 8.20 & 120.5792 & 2.8546 & \textbf{1.2558} & 96.02 & 2.27 \\
$(2048, 769, 16)$ & 19.6864 & 12.6251 & \textbf{1.2376} & 15.91 & 10.20 & 241.1584 & 11.8220 & \textbf{4.9970} & 48.26 & 2.37 \\
\bottomrule
\end{tabular}
}
\caption{Comparison of matrix multiplication algorithms (running time in seconds). $^\ast$For \cite{jin2020secure}, we report estimated running times due to memory issues. The actual number of (constant) multiplications  and rotations can be found in Appendix. $^\dagger$\cite{crockett2020low}.}

\label{tab:matmul_comp}
\end{table*}

\vspace{-1ex}
\subsection{Encrypted Matrix Multiplication}

Matrix multiplication accounted for a large portion of the total training time.
For example, when we fine-tuned the model on the CIFAR-10 dataset, the total running time for matrix multiplication took approximately 1712 seconds, which was more than 55\% of the total training time.
This explains the motivation
to develop efficient matrix multiplication algorithms.

We compared our matrix multiplication algorithms with those in \cite{jin2020secure} and \cite{crockett2020low}, and the results are listed in Table \ref{tab:matmul_comp}.
Owing to the memory limitation of the GPU, we could not measure the cost of the algorithms in \cite{jin2020secure},
therefore, we report the estimated costs by counting the number of multiplications and rotations.
The last three rows correspond to the shapes of the data and parameter matrices in our transfer learning experiments.
The table shows that our matrix multiplication algorithms are substantially faster than the baseline algorithms, by 1.8 to 323 times.
We emphasize the importance of minimizing both the number of rotations and multiplications in algorithms, even though a single multiplication takes longer than a single rotation.
To illustrate this point, consider the computation of $AB^\intercal$ with matrices of sizes $A \in \mathbb{R}^{2048 \times 769}$ and $B \in \mathbb{R}^{16 \times 769}$.
With the \textsf{ColMajor} algorithm, this operation requires 784 multiplications but 8703 rotations; hence, the rotation accounts for approximately 90\% of the total cost.
Our algorithm \textsf{DiagABT} significantly reduces the number of operations to 392 multiplications and 456 rotations and yields a speed increase by a factor of 10.2.
Table \ref{tab:matnum} in Appendix \ref{appendix:matrixmult} reports the actual numbers of each operation with various input matrix sizes.

\section{Related Work}

% \vspace{-1ex}
\paragraph{Softmax Approximation}
It is challenging to approximate the softmax function using polynomials in HE. Most previous works could not directly target the softmax but suggest alternative functions. The authors of \cite{al2020privft} used a quadratic polynomial approximation using the Minimax approximation algorithm. For the exponential function approximation, there was no difference in accuracy between Minimax and $L^2$-approximations.
The authors of \cite{lee2022privacy} used the Gumbel softmax function instead of the original softmax function to make the input be in the approximation region.
In \cite{hong2022secure}, the authors used the approximation $\mathsf{AE}_{r, L}(x) = ((2^{r} + x) / L)^{2^r}$ for the scaled exponential function and combined it with Goldschmidt's algorithm to obtain an approximation of softmax.
However, one must carefully select the scaling factor $L$, which is generally difficult.
PrivGD \cite{jin2020secure} trained a model with encrypted data but used one-vs-each softmax \cite{aueb2016one} instead and approximated it with a product of degree-3 $L^2$-approximations of the sigmoid function.

% \vspace{-1ex}
\paragraph{Encrypted Matrix Multiplication}
The algorithms proposed in \cite{crockett2020low, jin2020secure} for computing encrypted matrix multiplications of the forms $AB^{\T}$ and $A^{\T}B$ are different from ours.
The authors of \cite{jin2020secure} used three different types of packing; Row-majored packing (\textsf{RP}), Column-majored packing (\textsf{CP}), and Replicated packing (\textsf{REP}).
A weight matrix was packed with \textsf{REP}, and the input/label matrices were packed with \textsf{CP}.
Although the algorithms in \cite{jin2020secure} only consume one multiplicative depth, their computational complexity is significantly higher than that of ours.
In addition, their method required many blocks to pack a matrix. 
In \cite{crockett2020low}, the author introduced \textsf{ColMajor} and \textsf{RowMajor} algorithms to compute $AB^{\T}$ and $A^{\T}B$, respectively. 
These algorithms aim to extract and replicate rows/columns and view them as matrix-vector/vector-matrix multiplications.
We experimentally showed that the execution times of our algorithms are significantly smaller than that of \cite{crockett2020low}.

Some encrypted matrix multiplication algorithms \cite{jiang2018secure, huang2021more} are of the form $AB$, that is, they do not include transpose. 
As mentioned previously, these are unsuitable for our purpose because they add transpose operations for each training iteration.
In addition, we cannot use \cite{jiang2018secure} since the input and weight matrices cannot be packed into a single block.

\section{Conclusion}

In this study, we proposed HETAL, an efficient HE-based transfer learning algorithm that protects data privacy. We demonstrated its practicality through extensive experiments on five benchmark datasets.
We believe that HETAL can be applied to other domains such as speech classification.
In addition, our matrix multiplication and softmax approximation can be used for various other purposes, such as the encrypted inference of neural networks with softmax activations.
\section*{Acknowledgements}
The authors would like to thank the anonymous reviewers for their valuable comments and suggestions.
This work was supported by Institute of Information \& communications Technology Planning \& Evaluation (IITP) grant funded by the Korea government (MSIT) [No.2022-0-01047, Development of statistical analysis algorithm and module using homomorphic encryption based on real number operation]. Mun-Kyu Lee was also supported by IITP grant funded by the Korea government (MSIT) [No.RS-2022-00155915, Artificial Intelligence Convergence Innovation Human Resources Development (Inha University)].


\bibliography{refs}
\bibliographystyle{icml2023}

%%%%%%%%%%%%%%%%%%%%%%%%%%%%%%%%%%%%%%%%%%%%%%%%%%%%%%%%%%%%%%%%%%%%%%%%%%%%%%%
%%%%%%%%%%%%%%%%%%%%%%%%%%%%%%%%%%%%%%%%%%%%%%%%%%%%%%%%%%%%%%%%%%%%%%%%%%%%%%%
% APPENDIX
%%%%%%%%%%%%%%%%%%%%%%%%%%%%%%%%%%%%%%%%%%%%%%%%%%%%%%%%%%%%%%%%%%%%%%%%%%%%%%%
%%%%%%%%%%%%%%%%%%%%%%%%%%%%%%%%%%%%%%%%%%%%%%%%%%%%%%%%%%%%%%%%%%%%%%%%%%%%%%%
\newpage
\appendix
\onecolumn
\section{Softmax approximation}
\label{appendix:softmax}

\subsection{Theoretical error analysis}

% You can have as much text here as you want. The main body must be at most $8$ pages long.
% For the final version, one more page can be added.
% If you want, you can use an appendix like this one, even using the one-column format.
%%%%%%%%%%%%%%%%%%%%%%%%%%%%%%%%%%%%%%%%%%%%%%%%%%%%%%%%%%%%%%%%%%%%%%%%%%%%%%%
%%%%%%%%%%%%%%%%%%%%%%%%%%%%%%%%%%%%%%%%%%%%%%%%%%%%%%%%%%%%%%%%%%%%%%%%%%%%%%%

\begin{lemma}
\label{lem:softmaxlarge}
Let $\mathbf{x} = (x_1, \dots, x_c) \in \mathbb{R}^{c}$ be a vector with $\max_{1\leq j \leq c} x_j = m \geq 0$.
Assume that $x_i \leq -r$.
Then the $i$-th entry of the output of the softmax is bounded as
$$
0 < \sm(\mathbf{x})_i \leq \frac{1}{1 + e^r}.
$$
\end{lemma}
\begin{proof}
It is clear that the output is positive.
For the second inequality, we have
\begin{align*}
    \sm(x_1, \dots, x_c) &= \frac{e^{x_i}}{e^{x_1} + \cdots + e^{x_i} + \cdots + e^{x_c}} \\
    &\leq \frac{e^{x_i}}{e^{x_1} + \cdots + e^{x_{i-1}} + e^{x_i}} \\
    &\leq \frac{e^{-r}}{e^{x_1} + \cdots + e^{x_{i-1}} + e^{-r}} \qquad\cdots(1) \\
    &\leq \frac{e^{-r}}{1 + e^{-r}} = \frac{1}{1 + e^r}.
\end{align*}
(When we fix $x_1, \dots, x_{i-1}$, then the function $x_i \mapsto e^{x_i} / (e^{x_1} + \cdots + e^{x_i})$ is increasing and this proves (1).)
\end{proof}

\begin{lemma}
\label{lem:softmaxmvt}
For any $\mathbf{x}, \mathbf{y} \in \mathbb{R}^{c}$, we have
$$
\|\sm(\mathbf{x}) - \sm(\mathbf{y})\|_\infty \leq \frac{1}{2} \|\mathbf{x} - \mathbf{y}\|_\infty
$$
\end{lemma}
\begin{proof}
Let $g_i(t):= \sm((1-t)\mathbf{x} + t\mathbf{y})_i$ for $t\in\mathbb{R}$ and $1\leq i \leq c$.
By the mean value theorem, for each $i$ there exists $t_i \in (0, 1)$ such that
$$
\sm(\mathbf{y})_i - \sm(\mathbf{x})_i = \frac{g_i(1) - g_i(0)}{1 - 0} = g_i'(t_i) = (J_{\sm}(\mathbf{z}_i) (\mathbf{y} - \mathbf{x})^{\T})_i
$$
where $J_{\sm}$ is a Jacobian matrix of softmax and $\mathbf{z}_i:= (1-t_i)\mathbf{x} + t_i \mathbf{y}$.
By direct computation, one can check that the Jacobian is given by
$$
J_{\sm}(\mathbf{z}_i) = \begin{bmatrix} s_1 - s_1^2 & -s_1 s_2 & \cdots & -s_1 s_c \\
-s_2 s_1 & s_2 - s_2^2 & \cdots & -s_2 s_c \\
\vdots & \vdots & \ddots & \vdots \\
-s_c s_1 & -s_c s_2 & \cdots & s_c - s_c^2 
\end{bmatrix}
$$
where $(s_1, \dots, s_c) = \sm(\mathbf{z}_i)$.
Then we have
\begin{align*}
    |\sm(\mathbf{y})_i - \sm(\mathbf{x})_i| &= |\langle (-s_i s_1, -s_i s_2, \dots, s_i -s_i^2, \dots, -s_i s_c), \mathbf{y} - \mathbf{x} \rangle| \\
    & \leq \|(-s_i s_1, -s_i s_2, \dots, s_i - s_i^2, \dots, -s_i s_c)\|_1 \| \mathbf{y} - \mathbf{x}\|_\infty \\
    & = \left(\sum_{j=1}^{c} s_i s_j + (s_i - 2s_i^2)\right) \|\mathbf{y}- \mathbf{x}\|_\infty \\
    & = 2(s_i - s_i^2) \|\mathbf{y} -\mathbf{x}\|_\infty \leq \frac{1}{2} \|\mathbf{y} - \mathbf{x}\|_\infty
\end{align*}
where the last inequality follows from $\sum_{j=1}^{c} s_j = 1$ and $2(s_i - s_i^2) = \frac{1}{2} - 2(s_i - \frac{1}{2})^2$.
This proves the desired inequality.
\end{proof}

\begin{lemma}
\label{lem:softmaxtrunc}
Let $\mathbf{x} \in \mathbb{R}^{c}$ with $x_1 \geq x_2 \geq \cdots \geq x_c$ and $x_1 \geq 0$.
For $1\leq k \leq c$, define $\mathbf{x}_{:k} \in \mathbb{R}^{k}$ as $\mathbf{x}_{:k} = (x_1, \dots, x_k)$.
If $x_{k+1} < -r \leq x_{k}$, for $1 \leq i \leq k$ we have
$$
0 \leq \sm(\mathbf{x}_{:k})_i -  \sm(\mathbf{x})_i \leq \frac{c-k}{i(c-1+e^{r})}.
$$
In particular, we have $0 \leq \sm(\mathbf{x}_{:k})_{i} - \sm(\mathbf{x})_i \leq (c-1) / (c-1 + e^{r})$.
\end{lemma}
\begin{proof}
The left inequality is clear.
For the right one, we have
\begin{align*}
    \sm(\mathbf{x}_{:k})_i - \sm(\mathbf{x})_i &\leq \frac{e^{x_i}}{e^{x_1} + \cdots + e^{x_k}} - \frac{e^{x_i}}{e^{x_1} + \cdots + e^{x_{k}} + (c-k)e^{-r}} \\
    &= \frac{e^{x_i}}{e^{x_1} + \cdots + e^{x_k}} \cdot \frac{(c-k)e^{-r}}{e^{x_1} + \cdots + e^{x_{k}} + (c-k)e^{-r}}.
\end{align*}
The first term can be bounded as
$$
\sm(\mathbf{x}_{:k})_{i} \leq \sm(x_i, x_i, \dots, x_i, -r, \dots, -r)_{i} < \frac{1}{i},
$$
and the second term can be bounded as
$$
\frac{(c-k)e^{-r}}{e^{x_1} + \cdots + e^{x_{k}} + (c-k)e^{-r}} \leq \frac{(c-k)e^{-r}}{1 + (c-1)e^{-r}} = \frac{c-k}{c-1 +e^r}.
$$
(Here we use $x_1 \geq 0$ and $x_2, \dots, x_k \geq -r$.) This proves the inequality.
\end{proof}

From now, we assume that $d(\cdot) \in \mathscr{D}(\delta, r, R, LR)$ for some $\delta, r, R, L$ with $L > 0$ and for $i \geq 1$,
\begin{align*}
    d_0(x)&:= d(x) \\
    d_i(x)&:= L^i d \left(\frac{x}{L^i}\right) \\
    D_i(x)&:= (d_0 \circ d_1 \circ \dots \circ d_{i-1})(x).
\end{align*}
\begin{lemma}
\label{lem:dndervupperbound}
Let $d(\cdot) \in \mathscr{D}(\delta, r, R, LR)$. For all $i \geq 1$ and $x \in [-r, r]$, we have $0 \leq D_i'(x) \leq 1$.
\end{lemma}
\begin{proof}
Use induction on $i$. It is true for $i = 1$ by definition ($D_1(x) = d(x)$).
Assume that we have $0 \leq D_{i-1}(x) \leq 1$ for $x \in [-r, r]$.
Since $D_i(x) = D_{i-1}(d_{i-1}(x))$, we have
$$
    D_i'(x) = D_{i-1}'(d_{i-1}(x)) \cdot d_{i-1}'(x) = D_{i-1}'(d_{i-1}(x)) \cdot d'\left(\frac{x}{L^{i-1}}\right).
$$
Now combining the induction hypothesis with $d_{i-1}(x) \in [-r, r]$ and $x / L^{i-1} \in [-r, r]$ results that the above expression lies between 0 and 1.
\end{proof}

The following lemma gives a lower bound of $D_n(x)$ for $x \in [-r, r]$ that does not depend on $n$.

\begin{lemma}
\label{lem:dnlowerbound}
For all $n\geq 1$ and $x \in [0, r]$, we have
$$
D_n(x) \geq x - \frac{\delta L^2}{L^2 - 1} x^3.
$$
\end{lemma}
\begin{proof}
By mean value theorem, for each $i \geq 2$, we have
\begin{align*}
    D_{i-1}(x) - D_{i}(x) &= D_{i-1}(x) - D_{i-1}(d_{i-1}(x)) \\
    &= D_{i-1}'(c) (x - d_{i-1}(x))\qquad\qquad (\text{for some }d_{i-1}(x) < c < x) \\
    &\leq x - d_{i-1}(x) \\
    &= L^{i-1}\left( \frac{x}{L^{i-1}} - d\left(\frac{x}{L^{i-1}}\right)\right) \\
    &\leq L^{i-1} \cdot \delta \left(\frac{x}{L^{i-1}}\right)^{3} = \frac{\delta}{L^{2(i-1)}}
\end{align*}
where the first inequality follows from Lemma \ref{lem:dndervupperbound}.
Hence we get
\begin{align*}
    D_n(x) &= x - (x - D_1(x)) - (D_1(x) - D_2(x)) - \cdots - (D_{n-1}(x) - D_n(x)) \\
    &\geq x - \left(\delta x^3 + \frac{\delta}{L^2}x^3 + \cdots + \frac{\delta}{L^{2(n-1)}}x^3\right) \\
    &= x - \delta \frac{1 - L^{-2n}}{1 - L^{-2}} x^3 \\
    &\geq x - \frac{\delta L^2}{L^2 - 1}x^3.
\end{align*}
\end{proof}

Now we prove our main theorem.

\begin{theorem}
\label{thm:softmaxbound}
Let $d(\cdot) \in \mathscr{D}(\delta, r, R, LR)$ be a DEP and abuse a notation so that $d: \mathbb{R}^{c} \to \mathbb{R}^{c}$ is a function that applies $d : \mathbb{R} \to \mathbb{R}$  component-wise.
Let $p: [-R, R]^c \to [-R, R]^c$ be a polynomial approximation of $\sm$ with minimax error
$\|\sm(\mathbf{x}) - p(\mathbf{x})\|_{\infty} < \epsilon$.
Let $\mathsf{Amax}: [-L^n R, L^n R]^c \to [-L^n R, L^n R]^{c}$ be a given polynomial appproximation of max function satisfying $\mathsf{Amax}(\mathbf{x}) \leq \mathsf{max}(\mathbf{x})$ and $\|\mathsf{Amax}(\mathbf{x}) - \mathsf{max}(\mathbf{x})\|_{\infty}   \leq r$.
Define a normalized vector $\mathbf{x}' = \mathrm{Norm}(\mathbf{x}) = (x_1', \dots, x_c') \in \mathbb{R}^{c}$ as $x'_j =  x_j - m$ for $1\leq j \leq c$, where $m = \mathsf{Amax}(\mathbf{x})$.
For $p_n := p \circ D_n$ with $D_n := d_0 \circ \cdots \circ d_{n-1}$ and $d_i(x) := L^{i} d(x/L^i)$, if $\mathbf{x} \in [-\frac{1}{2}L^n R, \frac{1}{2}L^n R]^c$, we have
$$
\|\sm(\mathbf{x}) - p_n(\mathbf{x}')\|_\infty <  \beta(c,\delta, r, L, d) + \epsilon
$$
where
$$
\beta(c, \delta, r,L, d):= \frac{1}{1 + \frac{e^r}{c-1}} + \frac{1}{1 + \frac{e^{r - \delta L^2 r^3 / (L^2 - 1)}}{c-1}} + \frac{\delta r^3 L^2}{2 (L^2 - 1)}.
$$
\end{theorem}
Note that the upper bound does not depend on $n$.
\begin{proof}
One can assume that $x_1 \geq x_2 \geq \cdots \geq x_c$.
By the assumption on $\mathsf{ApproxMax}$, $x_{c} \leq m \leq x_{1}$ and $r  \geq x_{1}' = x_1 - m \geq 0$.
In other words, $\mathbf{x}' \in [-L^n R, r]^{c}$.
Since $\sm(\mathbf{x}') = \sm(\mathbf{x})$, we have.
% Let $D_n := d_0 \circ \cdots \circ d_{n-1}$.
We have
\begin{align*}
    \|\sm(\mathbf{x}) - p_n(\mathbf{x'})\|_\infty 
 &= \|\sm(\mathbf{x}') - p_n(\mathbf{x}')\|_\infty \\
 &\leq \|  \sm(\mathbf{x}') - \sm\circ D_n(\mathbf{x}')\|_\infty + \|\sm\circ D_n(\mathbf{x}') - p\circ D_n (\mathbf{x}')\|_\infty
\end{align*}

For $1 \leq i \leq c$, if $x_i' \leq -r $, then by Lemma \ref{lem:softmaxlarge}, we have
$$
0 < \sm(\mathbf{x}')_i \leq \frac{1}{1 +e^r}, \quad 0 < \sm(D_n(\mathbf{x}'))_i \leq \frac{1}{1 + e^{D_n(r)}}.
$$
(Second inequality follows from $D_n(x_1') \geq 0$ and $D_n(x_i') \leq -D_n(r)$.)
If $x_{i}' \geq -r$, let $k \in \{1, \dots, c\}$ be a minimal number such that $x_{k}' \geq -r$, so that $x_{c}' \leq \cdots \leq x_{k+1}' < -r \leq x_k' \leq \cdots \leq x_1'$.
Let $\mathbf{x}_{:k}':= (x_1', \dots, x_k') \in \mathbb{R}^{k}$.
Then 
\begin{align*}
    |\sm(\mathbf{x}')_i - \sm(D_n(\mathbf{x}'))_i| &\leq |\sm(\mathbf{x}')_i - \sm(\mathbf{x}_{:k}')_i| \\
    &+ |\sm(\mathbf{x}_{:k}')_i - \sm(D_n(\mathbf{x}_{:k}'))_i|\\
    &+ |\sm(D_n(\mathbf{x}_{:k}'))_i - \sm(D_n(\mathbf{x}'))_i|.
\end{align*}
By Lemma \ref{lem:softmaxtrunc}, the first and third terms are bounded above by $(c-1) / (c -1 + e^r)$ and $(c-1)/(c-1+ e^{D_n(r)})$, respectively.
The second term can be bounded using Lemma \ref{lem:softmaxmvt} and Theorem 1 of \cite{cheon2022efficient} as
\begin{equation*}
    |\sm(\mathbf{x}'_{:k})_i - \sm(D_n(\mathbf{x}'_{:k}))_i| \leq \frac{1}{2}\|\mathbf{x}'_{:k} - D_n(\mathbf{x}'_{:k})\|_{\infty} \leq \frac{\delta r^3 L^2}{2(L^2 - 1)}.
\end{equation*}
Hence we get
$$
|\sm(\mathbf{x}')_i - \sm(D_n(\mathbf{x}'))_i| \leq \beta(c, \delta,  r, L, d).
$$
Since $1 / (1 + e^{r}) \leq 1 / (1 + e^{r} / (c-1)) < \beta(c, \delta, r, L, d)$ and $1 / (1 + e^{D_n(r)}) \leq 1 / (1 + e^{D_n(r)} / (c-1)) \leq 1 / ({1 + e^{r - \delta L^2 r^3 / (L^2 - 1)}} / (c-1))  < \beta(c, \delta, r, L, d)$, we get

$$
\|\sm(\mathbf{x}') - \sm(D_n(\mathbf{x}'))\|_\infty \leq \beta(c, \delta, r, L, d).
$$

For the other term, since $D_n(\mathbf{x}') \in [-R, R]^c$, we have $\|\sm(D_n(\mathbf{x}')) - p(D_n(\mathbf{x}'))\|_{\infty, [-L^n R, L^n R]^c} < \epsilon$.
By combining these, we get the inequality.
\end{proof}

\subsection{Algorithm}

Based on our softmax approximation, Algorithm \ref{alg:softmax} computes the row-wise softmax of a matrix $M \in \mathbb{R}^{a \times c}$, which gives $\mathbf{P} \in \mathbb{R}^{a \times c}$ when given $M = \mathbf{X}\mathbf{W}^{\T}$, the probability matrix for each input in a minibatch.
Here $\mathsf{AExp}$ and $\mathsf{AInv}$ are approximated exponential and inverse functions with the algorithms in \cite{lee2022privacy}, with our modified parameters.
Also $\enc{M}$ stands for the encoding of a matrix $M$ with submatrices of size $s_0 \times s_1$ - see Appendix B for the details.
Note that the $\mathbf{X}\mathbf{W}^{\T}$ is tiled horizontally (assuming that the matrix $\mathbf{W}$ is zero-padded and tiled vertically), and we also want that the result of softmax is tiled horizontally, to apply our matrix multiplication algorithm.
Line 1--19 corresponds to the normalization of input that computes a row-wise max and subtracts it from the input.
The loop in line 21--27 is the domain extension step, where line 24--27 yields a high accuracy.
After computing the softmax (line 28--32), line 33--35 make the result matrix tiled horizontally.
For the comparison function, we followed the algorithm in \cite{cheon2020efficient}.
More precisely, we approximate $x \mapsto (\mathsf{sgn}(x) + 1) / 2$ on $[-1, 1]$ as a polynomial $f(g(g(x)))$ where
$$
f(x) = - \frac{5}{16}\left( x^7 - \frac{21}{5}x^5 + 7x^3 - 7x\right), \quad g(x) = -\frac{12860}{1024}\left(x^7 - \frac{25614}{12860}x^5 + \frac{16577}{12860}x^3 - \frac{4589}{12860}x\right),
$$
and get an approximation for $\mathsf{comp}(a, b) := (\mathsf{sgn}(a - b) + 1) / 2$ with $-1/2 \leq a, b \leq 1/2$.
Note that increasing the number of compositions of $f$ and $g$ gives a better approximation of the comparison.
However, we have experimentally found that $f(g(g(x))$ is enough for our experiments.

\begin{algorithm}
\caption{Row-wise softmax approximation}
\label{alg:softmax}
 \textbf{Input}: $\enc{\underline{M}}$, for $M \in \mathbb{R}^{a \times c}$, $c \leq s_{1}$, $c' = 2^{\lceil \log_{2}c\rceil}$, $R_{\mathsf{orig}}$ (original approximation range), $L$ (domain extension ratio), $n$ (domain extension index), $\mathsf{Acomp}(\cdot,\cdot)$ (homomorphic approximate comparison function), precise (boolean). \\
 \textbf{Output}: $\enc{\underline{\mathsf{ASoftmax}(M)}}$
\begin{algorithmic}[1]
    % \STATE $\enc{\underline{M_{\mathsf{avg}}}} = \RS(\enc{\underline{M_{0}}})$
    \STATE $R_{\mathsf{max}} = \lceil R_{\mathsf{orig}} \cdot L^{n} \rceil$
    \STATE $D_{c} = (d_{ij})$ where 
    $$
        d_{ij} = \begin{cases}
            1 & 0 \leq j < c \\
            0 & \text{otherwise}
        \end{cases} 
    $$
    \STATE $\enc{M'} = \enc{\underline{M}} \times \left(\frac{1}{2R_{\mathsf{max}}}\right)$
    \IF{$c \neq c'$}
        \STATE $D_{\mathsf{padmask}} = (m_{ij})$ where
        $$
            m_{ij} = \begin{cases}
                0 & 0 \leq j < c \\
                1/2 & \text{otherwise}
            \end{cases}
        $$
        \STATE $\enc{M'} = \enc{M'} - \enc{D_{\mathsf{padmask}}}$
    \ENDIF
    \STATE $D_{\mathsf{firstcol}} = (m_{ij}')$ where
    $$
        m_{ij}' = \begin{cases}
            1 & j=0 \\
            0 & \text{otherwise}
        \end{cases}
    $$
    \FOR{$j=0$ to $\log_{2}(c')$}
        \STATE $\enc{M_{\mathsf{rot}}} = \mathsf{Lrot}(\enc{M_{\mathsf{max}}}, 2^j)$
        \STATE $\enc{M_{\mathsf{comp}}} = \mathsf{Acomp}(\enc{M_{\mathsf{max}}}, \enc{M_{\mathsf{rot}}})$
        \STATE $\enc{M_{\mathsf{max}}} = \enc{M_{\mathsf{max}}} \odot \enc{M_{\mathsf{comp}}} + \enc{M_{\mathsf{rot}}} \odot (1 - \enc{M_{\mathsf{comp}}})$
    \ENDFOR
    \STATE $\enc{M_{\mathsf{max}}} = \enc{M_{\mathsf{max}}} \cdot \enc{D_{\mathsf{firstcol}}}$
    \FOR{$j=0$ to $\log_{2}(s_{1})$}
        \STATE $\enc{M_{\mathsf{max}}} = \enc{M_{\mathsf{max}}} + \mathsf{Rrot}(\enc{M_{\mathsf{max}}}, 2^i \cdot s_1)$
    \ENDFOR
    \STATE $\enc{M_{\mathsf{max}}} = \enc{M_{\mathsf{max}}} \times (2R_{\mathsf{max}})$
    
    % \STATE $\enc{{M_{\mathsf{avg}}}} = \enc{{M_{\mathsf{avg}}}} \odot (\frac{c'}{cs_1} D_c)$
    \STATE $\enc{{M_{\mathsf{norm}}}} = (\enc{\underline{M}} - \enc{{M_{\mathsf{max}}}})\odot \enc{D_c}$
    \STATE $B(x):= x - \frac{4x^{3}}{27R_{\mathsf{orig}}^{2}}$
    \FOR{$i=n-1$ to 0}
        \STATE $\enc{{M_{\mathsf{norm}}}} = L^{i} \odot B(\enc{{M_{\mathsf{norm}}}} \odot L^{-i})$
    \ENDFOR
    \IF{precise}
        \STATE $B_{\mathsf{inv}}(x):= x - \frac{4}{27} \frac{L^{2}(L^{2n} - 1)}{L^{2n}(L^2 - 1)} \left(\frac{x^3}{R_{\mathsf{orig}}^2} - \frac{x^5}{R_{\mathsf{orig}}^4}\right)$
        \STATE $\enc{M_{\mathsf{norm}}} = B_{\mathsf{inv}}(\enc{M_{\mathsf{norm}}})$
    \ENDIF
    \\
    % \STATE $\enc{\underline{M_{\mathsf{extn}}}} = $
    \STATE $M_{\mathsf{exp}} = \mathsf{AExp}(\enc{{M_{\mathsf{norm}}}})$
    \STATE $M_{\mathsf{exp}} = M_{\mathsf{exp}} \odot \enc{D_{c}}$
    \STATE $M_{\mathsf{expsum}} = \RS(M_{\mathsf{exp}})$
    \STATE $M_{\mathsf{Z}} = \mathsf{AInv}(M_{\mathsf{expsum}})$
    \STATE $M_{\sm} = M_{\mathsf{expsum}} \odot M_{\mathsf{Z}}$
    \FOR{$j=0$ to $\log_{2}(s_1 / c')$} 
        \STATE $M_{\sm} = M_{\sm} + \rrot(M_{\sm}, 2^{j}\cdot s_{1} \cdot c')$
    \ENDFOR
    \\
    \STATE $\enc{\underline{\mathsf{ASoftmax}(M)}} = M_{\sm}$
\end{algorithmic}
\end{algorithm}

\subsection{Comparison with previous approaches}
\label{appendix:softmaxcomparison}

The following Table \ref{tab:softmaxerr} shows the maximum and average errors of the softmax approximations including  \cite{lee2022privacy,hong2022secure,jin2020secure} and ours, for each input dimension $c$ and range $R$.
Since it is computationally intractable to find the exact maximum error of functions in several variables, we randomly sample points on each domain of approximation instead and report its maximum.
More precisely, according to the value of $R$, we sample as
\begin{itemize}
    \item $R=4$: sample 100M points uniformly on $[-4, 4]^{c}$,
    \item $R=8$: sample 100M points on $[-4, 4]^{c}$ and $[-8, 8]^{c}$ uniformly, total 200M points,
    \item $R=32$: sample 100M points on $[-4, 4]^{c}$, $[-8, 8]^{c}$, and $[-32, 32]^{c}$ uniformly, total 300M points.
    \item $R=128$: sample 100M points on $[-4, 4]^{c}$, $[-8, 8]^{c}$, $[-32, 32]^c$, and $[-128, 128]^{c}$ uniformly, total 400M points.
\end{itemize}
(we sample the points in such an accumulative way since uniformly randomly sampled points become more \emph{sparse} as $R$ increases, so we additionally sample points on smaller intervals to consider possible large error on small intervals, too.)
We can see that our approximation could cover the widest range with smallest error.
When the error is too large (e.g. Goldschmidt's algorithm fail to converge due to large input value), we filled up the corresponding entry with -.

For the comparison, we use the following parameters.
\begin{itemize}
    \item \cite{lee2022privacy}: We use the same parameters given in the paper. In particular, we use degree 12 $L^2$-approximation of the exponential function on $[-1, 1]$ with $B=64$, $R=10000$ and $n=8$ for inverse approximation, and Gumbel softmax function is used with $\lambda=4$.
    \item \cite{hong2022secure}: We use the same parameters given in the paper.
    In particular, we use $(r, L) = (4, 32)$ for $\mathsf{AExp}_{r, L}$ and $M=80$ and $d=30$ for inverse (which are the same as $R=80$ and $d=30$ if we use the notations from \cite{lee2022privacy}).
    \item \cite{jin2020secure}: They used degree 3 $L^2$-approximation of sigmoid on $[-8, 8]$  (that is $y = 0.5+ 0.15012x - 0.00159x^3$, used in \cite{kim2018secure}) which has a minimax error of 0.0098 on $[-8, 8]$.
    So we have observed the resulting one-vs-each softmax approximation error on $[-4, 4]^c$. (Since we take differences between inputs for each class, the actual input of each sigmoid could fit into $[-8, 8]$ when the input themselves are in $[-4, 4]$.)
    \item Ours: Our initial softmax approximation is based on \cite{lee2022privacy}, but with different parameters.
    We use $B = 8$ for exponential (approximation on $[-8, 8]$) and $R = 100$ and $n = 16$ for inverse. 
    Since normalization subtracts (approximate) maximum value from inputs, the possible range of the resulting normalized input becomes twice; hence our softmax can only cover $[-4, 4]^c$ without domain extension.
    For domain extension, we set the domain extension ratio as $L = 2$ and domain extension index as $5$, so that the new domain of approximation becomes $[-128, 128]^{c}$.
    Also, we can increase precision by applying Algorithm 2 of \cite{cheon2020efficient} which applies an additional degree 5 polynomial which approximates the inverse of DEP.
\end{itemize}

\begin{table}[ht]
\centering
\resizebox{\columnwidth}{!}{%
\setlength\extrarowheight{2pt}
% \begin{tabular}{@{\extracolsep{3pt}}*{9}{c}}
\begin{tabular}{c|c|c|c|c|c|c|c|c|c|c|c|c|c}
% \begin{tabular}{ccccc}
% {c|cccc|cc}
\toprule
$c$ & $R$ & \multicolumn{2}{c|}{\cite{lee2022privacy}} & \multicolumn{2}{c|}{\cite{hong2022secure}} & \multicolumn{2}{c|}{\cite{jin2020secure}} &  \multicolumn{2}{c|}{Ours (norm)} & \multicolumn{2}{c|}{Ours (norm+extn)} & \multicolumn{2}{c}{Ours (norm+extn+prec)} \\
\midrule
\multirow{4}{*}{3} & 4 & 0.9243 & 0.6957 & 0.0755 & 0.0162 & 0.1841 & 0.0910 & 3.7e-7 & 4.3e-8 & 0.0037 & 0.0015 & 0.0013 & 0.0003 \\
& 8 & 0.9651 & 0.7131 & 0.6041 & 0.0199 & - & - & - & - & 0.0037 & 0.0015 & 0.0013 & 0.0004 \\
& 32 & 0.9997 & 0.5812 & - & -& - &- & - & - & 0.0037 & 0.0015 & 0.0022 & 0.0006 \\
& 128 & - & - & - & -& - &- & - & - & 0.0037 & 0.0014 & 0.0022 & 0.0006 \\
\midrule
% \mrow{2}{5} & \\
% \cline{2-6} & \\
\multirow{4}{*}{5} & 4 & 0.9138 & 0.5784 & 0.0965 & 0.0239 & 0.3093 & 0.1148 & 7.2e-7 & 7.4e-8 & 0.0071 & 0.0029 & 0.0026 & 0.0004 \\
& 8 & 0.9591 & 0.5320 & 0.4121 & 0.0313 & - & - & - & - & 0.0071 & 0.0029 & 0.0026 & 0.0005 \\
& 32 & 0.9996 & 0.4806 & - & - & - & - & - & - & 0.0071 & 0.0024 & 0.0044 & 0.0008 \\
& 128 & - & - & - & - & - & - & - & - & 0.0116 & 0.0023 & 0.0044 & 0.0010 \\
\midrule
\multirow{4}{*}{7} & 4 & 0.9051 & 0.4926 & 0.1030 & 0.0297 & 0.4930 & 0.1296 & 1.0e-6 & 9.1e-8 & 0.0065 & 0.0031 & 0.0026 & 0.0003\\
& 8 & 0.9522 & 0.5320 & 0.2095 & 0.0416 & - & - & - & - & 0.0073 & 0.0029 & 0.0030 & 0.0005 \\
& 32 & 0.9992 & 0.4252 & - & - & - & - & - & - & 0.0087 & 0.0029 & 0.0065 & 0.0010  \\
& 128 & - & - & - & - & - & - & - & - & 0.0089 & 0.0029 & 0.0066 & 0.0013 \\
\midrule
\multirow{4}{*}{10} & 4 & 0.8942 & 0.3992 & 0.0948 & 0.0339 & 0.8018 & 0.1501 & 1.4e-6 & 1.0e-7 & 0.0153 & 0.0046 & 0.0040 & 0.0006 \\
& 8 & 0.9438 & 0.4476 & 0.2167 & 0.0516 & - & - & - & - & 0.0153 & 0.0042 & 0.0050 & 0.0014 \\
& 32 & 0.9985 & 0.3768 & - & - & - & - & - & - & 0.0153 & 0.0039 & 0.0094 & 0.0012 \\
& 128 & - & - & - & - & - & - & - & - & 0.0224 & 0.0039 & 0.0097 & 0.0016  \\
\bottomrule
\end{tabular}
}
\caption{Maximum and average errors of softmax approximation with 100--300M sampled points. Ours (norm+extn+prec) represents our approach combined with Algorithm 2 of \cite{cheon2022efficient}.}
\label{tab:softmaxerr}
\end{table}

The errors of previous works are fairly large, considering that softmax values lie between 0 and 1, and one can ask if it is possible to use these approximations in practice. The authors of \cite{lee-etal-2022-privacy}. (resp. \cite{hong2022secure}) proposed a softmax approximation and showed it is useful for ResNet-20 inference (resp. shallow neural network), where it is sufficient to identify the largest of many values rather than to calculate the exact values. This explains why their approximation works well for their purpose. However, for training, we need a softmax approximation that works well uniformly on large intervals; therefore, previous algorithms are not exactly suitable for training.
Although the one-vs-rest softmax is used for training in \cite{jin2020secure}, the input range is limited to $[-4, 4]$.

\subsection{Softmax input with vanilla SGD training}

\pgfplotstableread[col sep = comma]{plots2/cifar10_sgd.csv}\cifardata
\pgfplotstableread[col sep = comma]{plots2/mnist_sgd.csv}\mnistdata
\pgfplotstableread[col sep = comma]{plots2/facial-mask-detection_sgd.csv}\maskdata
\pgfplotstableread[col sep = comma]{plots2/dermamnist_sgd.csv}\dermadata
\pgfplotstableread[col sep = comma]{plots2/snips_sgd.csv}\snipsdata

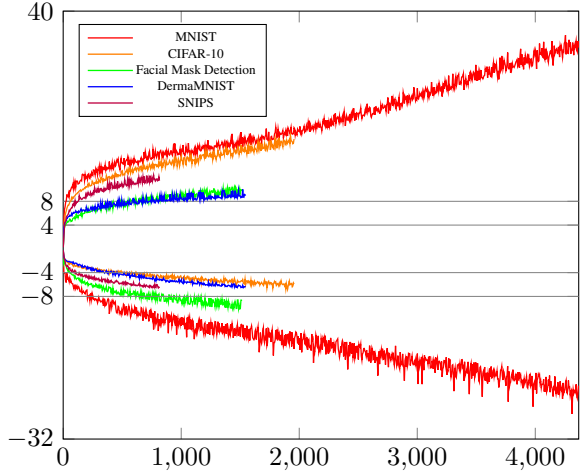
\begin{figure}
    \centering
    \begin{tikzpicture}
      \begin{axis}[
        legend pos = north west,
        legend style={nodes={scale=0.5, transform shape}},
        xmin = 0, xmax = 4368,
        ymin = -32, ymax = 40,
        ytick = {-32, -8,-4, 4, 8, 40}
        ]
        \addplot[line width = 0.5pt, red] table [x index = {0}, y index = {1}] {\mnistdata};
        \addplot[line width = 0.5pt, orange] table [x index = {0}, y index = {1}] {\cifardata};
        \addplot[line width = 0.5pt, green] table [x index = {0}, y index = {1}] {\maskdata};
        \addplot[line width = 0.5pt, blue] table [x index = {0}, y index = {1}] {\dermadata};
        \addplot[line width = 0.5pt, purple] table [x index = {0}, y index = {1}] {\snipsdata};
        \addplot[mark=none, gray, line width=0.1pt] coordinates {(0,8) (4368,8)};
        \addplot[mark=none, gray, line width=0.1pt] coordinates {(0,-8) (4368,-8)};
        \addplot[mark=none, gray, line width=0.1pt] coordinates {(0,4) (4368,4)};
        \addplot[mark=none, gray, line width=0.1pt] coordinates {(0,-4) (4368,-4)};
        \addplot[line width = 0.5pt, red] table [x index = {0}, y index = {2}] {\mnistdata};
        \addplot[line width = 0.5pt, orange] table [x index = {0}, y index = {2}] {\cifardata};
        \addplot[line width = 0.5pt, green] table [x index = {0}, y index = {2}] {\maskdata};
        \addplot[line width = 0.5pt, blue] table [x index = {0}, y index = {2}] {\dermadata};
        \addplot[line width = 0.5pt, purple] table [x index = {0}, y index = {2}] {\snipsdata};
        \legend{MNIST,CIFAR-10,Facial Mask Detection,DermaMNIST,SNIPS}
      \end{axis}
    \end{tikzpicture}
    \caption{Maximum and minimum value of input of softmax at each step (minibatch) for each dataset, where the model is trained with vanilla SGD.}
    \label{fig:softmaxinputsgd}
\end{figure}

Figure \ref{fig:softmaxinputsgd} shows how the minimum and maximum values of input of softmax vary as training proceeds when we use vanilla SGD instead of NAG.
Although the input values increase slower than that with NAG, the values are still significant and cannot be covered by the previous softmax approximation methods.
(Note that it took about 10 times longer than NAG to train models.)

\section{Encrypted matrix multiplication}

\subsection{Matrix composed of multiple blocks}
In the main article, we assumed that each matrix could fit into a single block (message or ciphertext) to simplify explanations.
% In this section, we give a detailed explanation that how the algorithm changes when 
Now we give a detailed description of the matrix multiplication algorithms for matrices whose encodings are composed of multiple blocks.

Let $A \in \mathbb{R}^{a \times b}$ be a matrix.
Fix a unit matrix shape $s_{0} \times s_{1}$, where $s = s_{0}s_{1}$ equals the number of slots in a single ciphertext.
When $a \leq s_{0}$ and $b  \leq s_{1}$, 
% we add zero-padding to the matrix $A$ to the right and bottom of $A$ 
we apply zero-padding to the right end and bottom of $A$ 
and encode it into a single block in a row-major manner. 
$\enc{A}$ denotes this encoding.
For example, if $a = b = 3$ and $s_{0} = s_{1} = 4$, the matrix $A = (a_{ij})_{0\leq i, j < 3}$ is encoded as
\begin{align*}
    \enc{A} = (a_{00}, a_{01}, a_{02}, 0, a_{10}, a_{11}, a_{12}, 0,  a_{20}, a_{21}, a_{22}, 0, 0, 0, 0, 0).
\end{align*}
If $a > s_{0}$ or $b > s_{1}$, we first zero-pad $A$ so that the number of rows and columns are multiples of $s_{0}$ and $s_{1}$, respectively, and then split $A$ into submatrices of shape $s_{0} \times s_{1}$. 
Then encoding each submatrix gives an encoding $\enc{A}$ of $A$,
$$
    \enc{A} = \{ \enc{A}_{i, j} \}_{0 \leq i < \lceil a / s_{0} \rceil, 0 \leq j < \lceil b / s_{1} \rceil},
$$
where $\enc{A}_{i, j}$ is the encoding of $(i, j)$-th submatrix.
(See also Figure \ref{fig:encode}.)
As explained in \cite{crockett2020low}, we can extend the $\RS$ and $\CS$ algorithm for large matrices.

\begin{figure}[ht]
    \centering
    \includegraphics[width=0.45\textwidth]{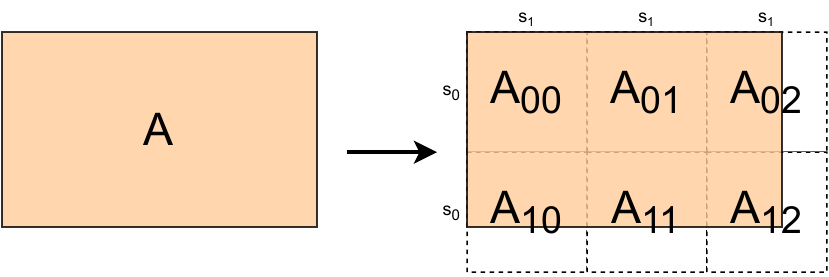}
    \caption{Encoding of a matrix $A \in \mathbb{R}^{13 \times 21}$ into 6 blocks where each encoded matrix of unit shape $8 \times 8$.}
    \label{fig:encode}
\end{figure}

\subsection{Proofs}

Here we give proofs for the propositions on encrypted matrix multiplication algorithms.

\begin{proposition}
\label{prop:abtappendix}
Let $A, B$ as above.
We have $A\overline{B}^{\T} = X + \conj(X)$ where
\small
\begin{equation}
X = \sum_{0 \leq k < c/2} \CS(A \odot \RU(\overline{B}_{\cplx}, k)) \odot M_{\cplx}^{(k, c)}.
\label{eqn:abt}
\end{equation}
\normalsize
Here $M^{(k, d)}$ is an off-diagonal masking matrix with entries
$$
M^{(k, d)}_{i, j} = 
    \begin{cases}
        1 & j \equiv i + k \,(\mathrm{mod}\, d) \\
        0 & \text{otherwise}
    \end{cases}
$$
and $M_{\cplx}^{(k, c)}$ is a \emph{complexified} version of the mask, which is
$$
M_{\cplx}^{(k, c)} = \frac{1}{2} M^{(k, c)} - \frac{\sqrt{-1}}{2} M^{(k + c/2, c)}
$$
\end{proposition}

\begin{proof}
We will first show that 
\begin{equation}
A\overline{B}^{\T} = \sum_{0 \leq k < c} \CS(A \odot \RU(\overline{B}, k)) \odot M^{(k, c)}.
\label{eqn:abtorig}
\end{equation}
It is enough to show that the $(i, j)$-th entry of the right-hand side equals
$\mathbf{a}_{i}\mathbf{b}_{j}^{\T}$, where $\mathbf{a}_{i}$ (resp. $\mathbf{b}_{j}$) is $i$-th (resp. $j$-th) row of $A$ (resp. $\overline{B}$).
Choose $0 \leq k_{0} < c$ such that $j - i \equiv k_{0} \,(\mathrm{mod}\,c)$.
Then all the $(i, j)$-th entries of summands of the right hand side vanishes 
% except for the $k_{0}$
except for the summand with index $k = k_0$ because of the masking.
For $k = k_{0}$, the $(i, j)$-th entry equals the dot product of the $i$-th row of $A$ and the $i$-th row of $\RU(B, k_0)$, and the latter
is $i + k_{0} \equiv j$-th row of $B$.

Now, we can see that 
\begin{align*}
\RU({\overline{B}_{\cplx}}, k) = \RU({\overline{B}}, k) + \sqrt{-1}\RU({\overline{B}}, k + c/2)
\end{align*}
and by the linearity of $\CS$ and bi-linearity of $\odot$, we get
\begin{align*}
    \CS(A\odot \RU(\overline{B}_{\cplx}, k)) =\CS(A\odot \RU(\overline{B}, k)) +\sqrt{-1}\CS(A \odot \RU(\overline{B}, k + c/2)).
\end{align*}
Now, combining this with equation,
$$
\Re((x + zi) (y - wi)) = xy + zw,\quad x, y, z, w \in \mathbb{R},
$$
we get
\begin{align*}
    &2\Re(\CS(A\odot \RU(\overline{B}_{\cplx}, k)) \odot M_{\cplx}^{(k, c)}) \\
    &= \CS(A\odot \RU(\overline{B}, k)) \odot M^{(k, c)} \\
    &+ \CS(A \odot \RU(\overline{B}, k + c/2)) \odot M^{(k + c/2, c)}.
\end{align*}
In other words, the $k$-th summand of Equation \eqref{eqn:abt} equals to the sum of the $k$-th and $(k + c/2)$-th summands of Equation \eqref{eqn:abtorig}, and this completes the proof.
\end{proof}

\begin{proposition}
\label{prop:atbdepthappendix}
$\underline{A}^{\T}B =  X + \mathsf{conj}(X)$, where 
\small
$$
X = \sum_{0 \leq k < c/2} \RS(\lrot(\underline{A}_{\cplx}, k) \odot \PRU(B, k)) \odot M_{\cplx}^{(-k, c)}. 
$$
\normalsize
\end{proposition}
\begin{proof}
Once we show the following identity
\begin{align*}
\RS(\lrot(\underline{A}_{\cplx}, k)\odot \PRU(B, k)) = \RS(\RL(\underline{A}_{\cplx}, k) \odot B),
\end{align*}
our equation is equivalent to 
$$
X = \sum_{0 \leq k < c/2} \RS(\RL(\underline{A}_{\cplx}, k) \odot B) \odot M_{\cplx}^{(-k, c)}. 
$$
which can be proved in a similar way as Proposition 2.
We can check that the first $(b - k)$ columns of $\RL^{*}(\underline{A}_{\cplx}, k)$ coincide with them of $\RL(\underline{A}_{\cplx},k)$,
and same thing holds for $\PRU(B, k)$ and $B$.
The last $k$ columns of $\RL^{*}(\underline{A}_{\cplx}, k) \odot \PRU(B, k)$ are
$$
\begin{bmatrix}
    x_{2,1}\cdot y_{2,b-k+1} & \dots & x_{2,k}\cdot y_{2,b} \\
    x_{3,1}\cdot y_{3,b-k+1} & \dots & x_{3,k}\cdot y_{3,b} \\
    \vdots & \ddots & \vdots \\
    x_{1,1}\cdot y_{1,b-k+1} & \dots & x_{1,k}\cdot y_{1,b} \\
\end{bmatrix}
$$
where $\underline{A}_{\cplx} = (x_{i, j})$ and $B = (y_{i, j})$, and the sums of entries in each column equal to
them of $\RL(\underline{A}_{\cplx}) \odot B$.
\end{proof}

\subsection{Algorithms}

We give detailed algorithms that we used for computing encrypted matrix multiplications.
It is worth noting that there are some restrictions on the shape of matrices and unit matrices for encoding.
For example, Algorithm \ref{alg:abt} requires that the number of rows $c$ of $B$ should satisfy $1 < c \leq s_{0}$. 
Hence we set the unit matrix shape $s_{0}, s_{1}$ to satisfy the restriction for the actual implementation.

We first briefly explain how the operations like addition, multiplication, $\RS$, $\CS$, $\lrot$, and $\rrot$ are extended to encodings composed of several blocks, i.e. when 
$$\enc{A} = \{\enc{A}_{i, j}\}_{0 \leq i < m, 0 \leq j < n}.$$
Addition and multiplication are simple.
Let $\enc{A_{1}} = \{\enc{A_{1}}_{i, j}\}_{0\leq i < m_{1}, 0 \leq j < n_{1}}$ and $\enc{A_{2}} = \{ \enc{A_{2}}_{i, j}\}_{0 \leq i < m_{2}, 0 \leq j < n_{1}}$.
If $m_1 = m_2$ and $n_1 = n_2$, we define addition and multiplication as 
\begin{align*}
    \enc{A_{1}} + \enc{A_{2}} &= \{ \enc{A_{1}}_{i,j} + \enc{A_{2}}_{i,j} \}_{0 \leq i < m_1, 0\leq j <n_1}, \\
    \enc{A_{1}} \odot \enc{A_{2}} &= \{ \enc{A_{1}}_{i,j} \odot \enc{A_{2}}_{i,j} \}_{0 \leq i < m_1, 0 \leq j < n_1}.
\end{align*}
We can also define addition and multiplication when $m_1 = 1$ (or $m_2 = 1$) and $n_1 = n_2$, or $m_1 = m_2$ and $n_1 = 1$ (or $n_2 =1$) by duplicating sub-encodings.

To compute $\RS(\enc{A})$, we first add the sub-encodings vertically and apply $\RS$ to each block to get 
$$\RS(\enc{A}) = \{ \RS(\sum_{0 \leq i < m} \enc{A}_{i, j}) \}_{0 \leq j < n}.$$
Similarly, we define $\CS(\enc{A})$ as
$$
\CS(\enc{A}) = \{ \CS(\sum_{0 \leq j < n} \enc{A}_{i, j})\}_{0 \leq i < m}.
$$
Finally, we define $\lrot$ and $\rrot$ for $\enc{A} = \{\enc{A}_{i, j}\}$ as 
\begin{align*}
    \lrot(\enc{A}, k)  &= \{ \lrot(\enc{A}_{i, j}, k) \}, \\
    \rrot(\enc{A}, k)  &= \{ \rrot(\enc{A}_{i, j}, k) \},
\end{align*}
and we extend $\RU, \RL, \PRU$ similarly.

The following algorithms (Algorithms 1 to 4) are the actual algorithms we use for implementation.

\begin{algorithm}
\caption{\textsf{DiagABT}: Homomorphic evaluation of $tAB^{\T}$}
\label{alg:abt}
\textbf{Input}: $\enc{A}$, $\enc{\overline{B}}$, for $A \in \mathbb{R}^{a \times b}, B \in \mathbb{R}^{c \times b}$, $1 < c \leq s_{0}$, and $t\in\mathbb{R}$ \\
\textbf{Output}: $\enc{t A \overline{B}^{\T}}$
\begin{algorithmic}[1]
    \STATE ${B_{\cplx}} = \enc{\overline{B}} + \sqrt{-1} \RU(\enc{\overline{B}}, c/2)$
    \FOR{$0 \leq k < \frac{c}{2}$}
        \STATE $B_{k} = \RU(B_{\cplx}, k)$
        \STATE $R_{k} = \enc{A} \odot B_{k}$
        \STATE $R_{k} = \CS(R_{k})$
        \STATE $R_{k} = R_{k} \odot {t M_{\cplx}^{(k, c)}}$
        % \ENDFOR
    \ENDFOR
    \STATE $X = \sum_{0 \leq k < c/2} R_{k}$
    \STATE $\enc{t A \overline{B}^{\T}} = X + \conj(X)$ \\
\end{algorithmic}
\end{algorithm}

\begin{algorithm}
\caption{$\RL(\enc{\underline{A}}, k)$}
\label{alg:rl}
 \textbf{Input}: $\enc{{A}}$ where $A \in \mathbb{R}^{a \times s_{1}}$, $0 \leq k < s_{1}$ \\
 \textbf{Output}: $\RL(\enc{{A}}, k)$
\begin{algorithmic}[1]
    \STATE $D_{k} = (d_{ij})$ where 
    $$
    d_{ij} = \begin{cases} 1 & 0 \leq j < s_{1} - k \\ 0 & \text{otherwise}\end{cases}
    $$
    \STATE $A_{1} = \lrot(\enc{A}, k)$
    \STATE $A_{2} = A_{1} \odot \enc{D_{k}}$ \\
    \STATE $\RL(\enc{A}, k) = A_{2} + \rrot(A_{1} - A_{2}, s_{1})$
\end{algorithmic}
\end{algorithm}

\begin{algorithm}
\caption{$\PRU({B}, k)$}
\label{alg:pru}
 \textbf{Input}: $\langle {B} \rangle$ for $B \in \mathbb{R}^{a \times b}$, $0 \leq k < s_{1}$ \\
 \textbf{Output}: $\PRU({B}, k)$
\begin{algorithmic}[1]
    \STATE $D_{k} = (d_{ij})$ where 
    $$
    d_{ij} = \begin{cases} 1 & 0 \leq j < s_{1} - k \\ 0 & \text{otherwise}\end{cases}
    $$
    \STATE $B' = \enc{B} \odot \enc{D_{k}}$ \\
    \STATE $\PRU(B, k) = B' + \RU(\enc{B} - B', 1)$
\end{algorithmic}
\end{algorithm}

\begin{algorithm}
\caption{\textsf{DiagATB}: Homomorphic evaluation of $tA^{\T}B$}
\label{alg:atbcplxlowdepth}
 \textbf{Input}: $\enc{\overline{A}}$, $\enc{B}$,
for $A \in \mathbb{R}^{a \times c}, B \in \mathbb{R}^{a \times b}$, and $t\in\mathbb{R}$ \\
 \textbf{Output}: $\enc{t\overline{A}^{\T}B}$
\begin{algorithmic}[1]
    \STATE ${\overline{A}_{\mathsf{cplx}}} = \enc{\overline{A}} + \sqrt{-1} \RL(\enc{\overline{A}}, c/2)$
    \FOR{$0 \leq k < \frac{c}{2}$}
        \IF{$\mathsf{level}(A) < \mathsf{level}(B)$}
            \STATE $A_{k} = \lrot(\overline{A}_{\cplx}, k)$
            \STATE $B_{k} = \PRU(\enc{B}, k)$
            \STATE $R_{k} = A_{k} \odot B_{k}$
        \ELSE
            \STATE $A_{k} =  \RL(\overline{A}_{\cplx}, k)$
            \STATE $R_{k} = A_{k} \odot \enc{B}$
        \ENDIF
        \STATE $R_{k} = \RS(R_{k})$
        \STATE $R_{k} = R_{k} \odot {t M_{\cplx}^{(-k, a)}}$
    \ENDFOR
    \STATE $X = \sum_{0 \leq k < b/2} R_{k}$
    \STATE $\enc{t\overline{A}^{\T}B} = X + \conj(X)$ \\
\end{algorithmic}
\end{algorithm}

\section{Experiments}

\subsection{Dataset description}
\begin{itemize}
    \item \textbf{MNIST} \cite{deng2012mnist} is one of the most widely used image classification dataset, consisting of 70k images of handwritten digits, from 0 to 9.
    \item \textbf{CIFAR-10} \cite{krizhevsky2009learning} is another famous image classification dataset, consisting of 60k color images of 10 classes: airplane, automobile, bird, cat, deer, dog, frog, horse, sheep, truck.

    \item \textbf{Face Mask Detection} \cite{larxel2020face} is a dataset from Kaggle that contains 853 images with several peoples.
    Each face is classified as one of the following three: wearing a mask correctly, wearing a mask incorrectly, and not wearing a mask.
    With given metadata on each image, we crop faces and make them into single images, which results in total 4072 images.

    \item \textbf{DermaMNIST} \cite{yang2023medmnist} is one of the MedMNIST collection, which is a medical dataset of 10015 common pigmented skin lesions images based on the HAM10000 dataset \cite{tschandl2018ham10000}, where each image is labeled as one of the 7 diseases.

    \item \textbf{SNIPS} \cite{coucke2018snips} is a dataset of crowd-sourced queries collected from Snips Voice Platform, distributed along 7 user intents.
\end{itemize}

Table \ref{tab:datasize} describes the number of samples in each split (train, validation, test) for each benchmark.
The splits are already given for DermaMNIST and SNIPS datasets, and we randomly split original train sets into train and validation sets for the other datasets of the ratio 7:1.
We used these splits to find hyperparameters and report the final performances (execution time and model accuracy) in Table 2 of the main article.

\begin{table}[ht]
\centering
    % \resizebox{\columnwidth}{!}{%
    \begin{tabular}{c|cccc}
        \toprule
            Dataset & Train & Validation & Test & Total \\
        \midrule
            MNIST & 52500 & 7500 & 10000 & 70000\\
            CIFAR-10 & 43750 & 6250 & 10000 & 60000\\
            Face Mask Detection & 2849 & 408 & 815 & 4072 \\
            DermaMNIST & 7007 & 1003 & 2005 & 10015\\
            SNIPS & 13084 & 700 & 700 & 14484 \\
        \bottomrule
    \end{tabular}
    % }
    \caption{Number of samples in each benchmark dataset.}
    \label{tab:datasize}
\end{table}

\subsection{Hyperparameters}
Table \ref{tab:hyperparam} shows a list of hyperparameters (minibatch sizes and learning rates) that are used for experiments.
% {\color{blue}
For early-stopping, we set patience as 3 so that the server trains until the validation loss does not decrease further for 3 epochs.
% }

\begin{table}[ht]
\centering
    % \resizebox{\columnwidth}{!}{%
    \begin{tabular}{c|cccc}
        \toprule
            Dataset & total epochs & best epoch & batch size & learning rate \\
        \midrule
            MNIST & 7 & 4 & 1024 & 2.0\\
            CIFAR-10 & 9 & 6 & 2048 & 1.0 \\
            Face Mask Detection & 22 & 19 & 512 & 0.5 \\
            DermaMNIST & 23 & 20 & 1024 & 0.3 \\
            SNIPS & 14 & 11 & 1024 & 1.0 \\
        \bottomrule
    \end{tabular}
    % }
    \caption{Batch size, learning rate, and number of epochs (early-stopped and best) for each benchmark dataset.}
    \label{tab:hyperparam}
\end{table}

% \newpage
\subsection{Encrypted matrix multiplication}
\label{appendix:matrixmult}

First of all, for the comparison of encrypted matrix multiplication algorithms (Table 4 of the main article), we set $s_{0} = a$ for all experiments (our $\mathsf{DiagABT}, \mathsf{DiagATB}$ algorithms and $\mathsf{ColMajor}, \mathsf{RowMajor}$ of \cite{crockett2020low}).
The Table \ref{tab:matmulcomplexity} shows the computational complexity of each algorithm, and
Table \ref{tab:matnum} shows the actual number of constant multiplications, multiplications, and rotations used for each encrypted matrix multiplication algorithm.

\begin{table*}[ht]
\centering
\resizebox{\columnwidth}{!}{%
\setlength\extrarowheight{2pt}
\setlength\tabcolsep{1.5pt}
% \begin{tabular}{@{\extracolsep{5pt}}*{7}{c}}
% \begin{tabular}{@{\extracolsep{3pt}}*{9}{c}}
\begin{tabular}{c|c|c|c|c|c|c}
% \begin{tabular}{ccccc}
% {c|cccc|cc}
\toprule
\multirow{2}{*}{Ops}& \multicolumn{3}{c|}{$AB^{\T}$ ($A \in \mathbb{R}^{a\times b}, B \in \mathbb{R}^{c \times b}$)} & \multicolumn{3}{c}{$A^{\T}B$ ($A \in \mathbb{R}^{a \times c}, B \in \mathbb{R}^{a \times b}$)} \\
% \cline{2-5}  \cline{6-9}
\cline{2-7}
 & \cite{jin2020secure}$^{*}$ & \textsf{ColMajor} & \textsf{DiagABT} &  \cite{jin2020secure}$^{*}$ & \textsf{RowMajor} & \textsf{DiagATB}  \\
\midrule
$\mathsf{CMult}$ &  0 & $O(c(\frac{a}{s_0}+ \frac{b}{s_1}))$ & $O(\frac{ac}{2s_0})$ & 0 & $O(c(\frac{a}{s_0} + \frac{b}{s_1}))$ & $O(\frac{abc}{2s})$ \\
$\mathsf{Mult}$ & $O(bc)$ & $O(\frac{abc}{s})$ & $O(\frac{abc}{2s})$  & $O(bc)$ & $O(\frac{abc}{s})$ & $O(\frac{abc}{2s})$\\
$\mathsf{Rot}$ & 0 & $O(c (\frac{a}{s_0}\log s_1 + \frac{b}{s_1} \log s_0))$ & $O(c(\frac{a}{s_0}\log s_1 + \frac{b}{2s_1}) )$ & $O(bc \log s)$ &$O(c (\frac{a}{s_0}\log s_1 + \frac{b}{s_1} \log s_0))$ & $O(c (\frac{ab}{2s} + \frac{b}{2s_1}\log s_0))$\\
\bottomrule
\end{tabular}
}
\caption{Complexity of matrix multiplication algorithms. Note that $s = s_0 s_1$.}

\label{tab:matmulcomplexity}
\end{table*}

\begin{table*}[ht]
\centering
% \resizebox{\columnwidth}{!}{%
\setlength\extrarowheight{2pt}
% \begin{tabular}{@{\extracolsep{3pt}}*{9}{c}}
\begin{tabular}{c|ccc|ccc}
% \begin{tabular}{ccccc}
% {c|cccc|cc}
\toprule
\multirow{2}{*}{$(a, b, c)$}& \multicolumn{3}{c|}{$AB^{\T}$ ($A \in \mathbb{R}^{a\times b}, B \in \mathbb{R}^{c \times b}$)} & \multicolumn{3}{c}{$A^{\T}B$ ($A \in \mathbb{R}^{a \times c}, B \in \mathbb{R}^{a \times b}$)} \\
% \cline{2-5}  \cline{6-9}
\cline{2-7}
 & \cite{jin2020secure}$^{*}$ & \textsf{ColMajor} & \textsf{DiagABT} &  \cite{jin2020secure}$^{*}$ & \textsf{RowMajor} & \textsf{DiagATB} \\
\midrule
\mrow{3}{$(128, 128, 4)$} & 0 & 4 & 0 & 0 & 4 & 2 \\
\cline{2-7}
 & 512 & 4 & 2 & 512 & 4 & 0 \\
\cline{2-7}
 & 0 & 63 & 34 & 7680 & 63 & 18\\
\midrule
\mrow{3}{$(256, 256, 8)$} & 0 & 16 & 0 & 0 & 8 & 4\\
\cline{2-7}
& 2048 & 16 & 8 & 2048 &16 & 0\\
\cline{2-7}
& 0 & 191 & 64 & 30720 & 191 & 72 \\
\midrule
\mrow{3}{$(512, 769, 4)$} &0 & 52 & 0 &0 & 4 & 2\\
\cline{2-7}
& 3076 & 52 & 26 & 3076 & 52 & 0\\
\cline{2-7}
& 0 & 495 & 50 & 46140 & 495 & 238\\
\midrule
\mrow{3}{$(1024, 769, 8)$} & 0 & 200 & 0 & 0 & 8 & 4 \\
\cline{2-7}
& 6152 & 200 & 100 & 6152 & 200 & 0  \\
\cline{2-7}
& 0 & 2047 & 140 & 92280 & 2047 & 1008 \\
\midrule
\mrow{3}{$(2048, 769, 16)$} & 0& 784 & 0 & 0& 16 & 8 \\
\cline{2-7}
& 12304 & 784 & 392 & 12304 & 784 & 0\\
\cline{2-7}
& 0 & 8703 & 456 & 184560 & 8703 & 4328 \\
\bottomrule
\end{tabular}
% }
\caption{The number of constant multiplications (\textsf{CMult}, first rows), multiplications (\textsf{Mult}, second rows),  and rotations (\textsf{Rot}, third rows).}
\label{tab:matnum}
\end{table*}

\subsection{Using larger pre-trained models}
\label{appendix:largermodel}

We also conducted experiments with larger pre-trained models.
Especially, we replace the \texttt{ViT-Base} model for the image dataset with the \texttt{ViT-Large} model and see how the performance changes.
The hidden dimensions of the models are 768 and 1024, respectively, and the other information on the architectures of the models can be found in \cite{dosovitskiy2021an}.
The overall results with these larger models can be found in Table \ref{tab:largevit}, which shows that we can still apply HETAL and fine-tune the models with encrypted data in a reasonable amount of time (in 1.2 hours). The results from these experiments illustrate that HETAL is flexible and scales well with larger models.

We used the same minibatch sizes as in Table \ref{tab:hyperparam}, and also set patience as 3 for early-stopping.
The list of learning rates and number of epochs for each experiment can be found in Table \ref{tab:hyperparam-large}.

\begin{table*}[t]
\centering
\setlength\extrarowheight{2pt}
% \resizebox{2\columnwidth}{!}{%
\begin{tabular*}{\textwidth}{@{\extracolsep{\fill}}*{7}{c}}
\toprule
\multirow{3}{*}{dataset} & \multirow{3}{*}{model} & \multicolumn{3}{c}{encrypted} & \multicolumn{1}{c}{unencrypted} \\
\cline{3-5}  \cline{6-6}
& &  \multicolumn{2}{c}{Running time} & \multirow{2}{*}{ACC (a)} &  \multirow{2}{*}{ACC (b)} &\multirow{2}{*}{ACC loss ((b) - (a))}\\
\cline{3-4}
& & Total (s) & Time / Iter (s) & & &  \\
\cline{1-7}
\multirow{2}{*}{MNIST} & \texttt{Base} & 3442.29 & 9.46 & 96.73\% & 97.24\% & 0.51\%\\
 & \texttt{Large} & 4159.60 & 11.43 & 97.46\% & 98.13\% & 0.67\% \\
\midrule
\multirow{2}{*}{CIFAR-10} & \texttt{Base} & 3114.30 & 15.72 & 96.57\% & 96.62\% & 0.05\% \\
 & \texttt{Large} & 3073.06 & 19.95 & 97.36\% & 97.39\% & 0.03\% \\
\midrule
\multirow{2}{*}{Face Mask Detection} & \texttt{Base} & 566.72 & 4.29 & 95.46\% & 95.46\% & 0.00\% \\
 & \texttt{Large} & 347.94 & 5.80 & 95.34\% & 95.34\% & 0.00\% \\
\midrule
\multirow{2}{*}{DermaMNIST} & \texttt{Base} & 1136.99 & 7.06 & 76.06\% & 76.01\% & -0.05\% \\
 & \texttt{Large} & 879.27 & 8.37 & 76.86\% & 76.76\% & -0.10\%  \\
\bottomrule
\end{tabular*}
% }
\caption{HETAL with different sizes of ViTs.}
\label{tab:largevit}
\end{table*}

\begin{table}[ht]
\centering
    % \resizebox{\columnwidth}{!}{%
    \begin{tabular}{cc|cccc}
        \toprule
            Dataset & model & total epochs & best epoch & batch size & learning rate \\
        \midrule
            \multirow{2}{*}{MNIST} & \texttt{Base} & 7 & 4 & \multirow{2}{*}{1024} & 2.0\\
            & \texttt{Large} & 7 & 4 & & 0.05 \\
            \midrule
            \multirow{2}{*}{CIFAR-10} & \texttt{Base} & 9 & 6 & \multirow{2}{*}{2048} & 1.0 \\
            & \texttt{Large} & 7 & 4 & & 0.1 \\
            \midrule
            \multirow{2}{*}{Face Mask Detection} & \texttt{Base} & 22 & 19 & \multirow{2}{*}{512} & 0.5 \\
            & \texttt{Large} & 10 & 7 &  & 0.1  \\
            \midrule
            \multirow{2}{*}{DermaMNIST} & \texttt{Base} & 23 & 20 & \multirow{2}{*}{1024} & 0.3 \\
             & \texttt{Large} & 15 & 12 & & 0.03 \\
        \bottomrule
    \end{tabular}
    % }
    \caption{Batch size, learning rate, and number of epochs (early-stopped and best) for each benchmark dataset and model size.}
    \label{tab:hyperparam-large}
\end{table}

\subsection{Comparison between encrypted and unencrypted training}

\begin{table}[ht]
\centering
\begin{tabular}{l|ccc}
    \toprule
    Dataset & encrypted (s) & epochs & unencrypted (s) \\
    \midrule
    MNIST & 3442 & 14 & 194 \\
    CIFAR-10 & 3114 & 10 & 113 \\
    Face Mask Detection & 567 & 22 & 22 \\
    DermaMNIST & 1137 & 23 & 41 \\
    SNIPS & 1264 & 25 & 84 \\
    \bottomrule
\end{tabular}
\caption{Comparison of total runtime for encrypted and unencrypted training across various datasets.}
\label{tab:runtime}
\end{table}

We ran HETAL on the unencrypted datasets and compared the runtimes with those for encrypted datasets in Table \ref{tab:runtime}.
It is important to note that we implemented the fine-tuning module of HETAL for unencrypted data using NumPy \cite{harris20numpy} from scratch for a fair comparison, and the results are obtained without using a GPU.

Though the runtimes for encrypted training are longer, it is crucial to highlight that we have achieved practical performance levels with our homomorphic encryption implementation. The experimental results demonstrate that the training was completed in less than an hour for all five datasets with a dimension of 768, reinforcing the practical feasibility of HETAL.

We remark that both the unencrypted and encrypted versions could be further improved if additional optimizations were implemented.

\clearpage
\newpage
\mbox{}
% \clearpage
% \newpage

\section{Updates after publication}

Here, we document revisions to our paper subsequent to its official publication at ICML 2023. We extend our gratitude to those who have provided valuable feedback. Any necessary updates to the codes and \texttt{.whl} files on GitHub (https://github.com/CryptoLabInc/HETAL) will also be made available.

\subsection{Correction of the Table \ref{tab:matnum}}

We have made revisions to Table \ref{tab:matnum}, now referenced as Table \ref{tab:matnum_fixed}. During the creation of Table \ref{tab:matnum}, we identified an oversight in the inclusion of counters within function calls for basic operations, resulting in the omission of lines tallying \textsf{CMult} and \textsf{Mult}. This discrepancy has been rectified, and the counts for \textsf{CMult} and \textsf{Mult} across all four algorithms (\textsf{ColMajor}, \textsf{RowMajor}, \textsf{DiagABT}, and \textsf{DiagATB}) have been updated accordingly, as depicted in Table \ref{tab:matnum_fixed}. The revised code has been integrated into the latest version of the \texttt{.whl} files on the GitHub repository. Notably, Table \ref{tab:matmul_comp} remains accurate.

Furthermore, in the experimentation for Table \ref{tab:matmul_comp}, both levels of matrices $A$ and $B$ were set to $12$, representing the maximum achievable level within our CKKS parameter. Consequently, the \textsf{DiagATB} algorithm did not utilize \textsf{PRotUp} for depth optimization during the comparison. As discussed in the paper, when $\mathsf{level}(A) < \mathsf{level}(B)$, a depth can be conserved by one, albeit with an increase in the number of operations, as detailed in Table \ref{tab:matnum_fixed}. Additionally, Table \ref{tab:matmul_cnt_formula} has been introduced, offering precise operation counts for a given shape and $s_0, s_1$, supplementing the computational complexity presented in Table \ref{tab:matmulcomplexity}.

We express our gratitude to Miran Kim for bringing these discrepancies to our attention.

\subsection{Proof of Theorem \ref{thm:softmaxbound}}

In the proof of Theorem \ref{thm:softmaxbound}, we implicitly assumed that $D_n(r)$ is an odd function.
This is true for the actual polynomials we used for domain extension (the same polynomials as in \cite{cheon2022efficient}).
We thank Hojune Shin for pointing this out.

\begin{table*}[ht]
\centering
% \resizebox{\columnwidth}{!}{%
\setlength\extrarowheight{2pt}
% \begin{tabular}{@{\extracolsep{3pt}}*{9}{c}}
\begin{tabular}{c|ccc|ccc}
% \begin{tabular}{ccccc}
% {c|cccc|cc}
\toprule
\multirow{2}{*}{$(a, b, c)$}& \multicolumn{3}{c|}{$AB^{\T}$ ($A \in \mathbb{R}^{a\times b}, B \in \mathbb{R}^{c \times b}$)} & \multicolumn{3}{c}{$A^{\T}B$ ($A \in \mathbb{R}^{a \times c}, B \in \mathbb{R}^{a \times b}$)} \\
% \cline{2-5}  \cline{6-9}
\cline{2-7}
 & (Jin et al.)  & \textsf{ColMajor} & \textsf{DiagABT} &  (Jin et al.)  & \textsf{RowMajor} & \textsf{DiagATB} \\
\midrule
\mrow{5}{$(128, 128, 4)$} & \mrow{2}{0} & \mrow{2}{8} & \mrow{2}{4} & \mrow{2}{0} & \mrow{2}{8} & 4 \\
\cline{7-7}
 &  &  &  &  &  & 4 \\
\cline{2-7}
 & 512 & 4 & 2 & 512 & 4 & 2 \\
\cline{2-7}
 & \mrow{2}{0} & \mrow{2}{63} & \mrow{2}{34} & \mrow{2}{7680} & \mrow{2}{63} & 18\\
\cline{7-7}
 &  &  &  &  &  & 18 \\
\midrule
\mrow{5}{$(256, 256, 8)$} & \mrow{2}{0} & \mrow{2}{24} & \mrow{2}{8} & \mrow{2}{0} & \mrow{2}{24} & 12 \\
\cline{7-7}
 &  &  &  &  &  & 15 \\
\cline{2-7}
& 2048 & 16 & 8 & 2048 &16 & 8\\
\cline{2-7}
& \mrow{2}{0} & \mrow{2}{191} & \mrow{2}{64} & \mrow{2}{30720} & \mrow{2}{191} & 72 \\
\cline{7-7}
 &  &  &  &  &  & 75 \\
\midrule
\mrow{5}{$(512, 769, 4)$} & \mrow{2}{0} & \mrow{2}{56} & \mrow{2}{4} & \mrow{2}{0} & \mrow{2}{56} & 28\\
\cline{7-7}
 &  &  &  &  &  & 40 \\
\cline{2-7}
& 3076 & 52 & 26 & 3076 & 52 & 26\\
\cline{2-7}
& \mrow{2}{0} & \mrow{2}{495} & \mrow{2}{50} & \mrow{2}{46140} & \mrow{2}{495} & 238\\
\cline{7-7}
 &  &  &  &  &  & 250 \\
\midrule
\mrow{5}{$(1024, 769, 8)$} & \mrow{2}{0} & \mrow{2}{208} & \mrow{2}{8} & \mrow{2}{0} & \mrow{2}{208} & 104 \\
\cline{7-7}
 &  &  &  &  &  & 176 \\
\cline{2-7}
& 6152 & 200 & 100 & 6152 & 200 & 100  \\
\cline{2-7}
& \mrow{2}{0} & \mrow{2}{2047} & \mrow{2}{140} & \mrow{2}{92280} & \mrow{2}{2047} & 1008 \\
\cline{7-7}
 &  &  &  &  &  & 1080 \\
\midrule
\mrow{5}{$(2048, 769, 16)$} & \mrow{2}{0} & \mrow{2}{800} & \mrow{2}{16} & \mrow{2}{0} & \mrow{2}{800} & 400 \\
\cline{7-7}
 &  &  &  &  &  & 736 \\
\cline{2-7}
& 12304 & 784 & 392 & 12304 & 784 & 392 \\
\cline{2-7}
& \mrow{2}{0} & \mrow{2}{8703} & \mrow{2}{456} & \mrow{2}{184560} & \mrow{2}{8703} & 4328 \\
\cline{7-7}
 &  &  &  &  &  & 4664 \\
\bottomrule
\end{tabular}
% }
\caption{The number of constant multiplications (\textsf{CMult}, first rows), multiplications (\textsf{Mult}, second rows),  and rotations (\textsf{Rot}, third rows).
For \textsf{DiagATB}, the first row (resp. second row) of \textsf{CMult} and \textsf{Rot} represent the case when $\mathsf{level}(A) \geq \mathsf{level}(B)$ (resp. $\mathsf{level}(A) < \mathsf{level}(B)$).}
\label{tab:matnum_fixed}
\end{table*}

\begin{table*}[ht]
\centering
\resizebox{\columnwidth}{!}{%
\setlength\extrarowheight{2pt}
\setlength\tabcolsep{1.5pt}
% \begin{tabular}{@{\extracolsep{5pt}}*{7}{c}}
% \begin{tabular}{@{\extracolsep{3pt}}*{9}{c}}
\begin{tabular}{c|c|c|c|c|c|c}
% \begin{tabular}{ccccc}
% {c|cccc|cc}
\toprule
\multirow{2}{*}{Ops}& \multicolumn{3}{c|}{$AB^{\T}$ ($A \in \mathbb{R}^{a\times b}, B \in \mathbb{R}^{c \times b}$)} & \multicolumn{3}{c}{$A^{\T}B$ ($A \in \mathbb{R}^{a \times c}, B \in \mathbb{R}^{a \times b}$)} \\
% \cline{2-5}  \cline{6-9}
\cline{2-7}
 & (Jin et al.)  & \textsf{ColMajor} & \textsf{DiagABT} &  (Jin et al.)  & \textsf{RowMajor} & \textsf{DiagATB}  \\
\midrule
\mrow{2}{$\mathsf{CMult}$} & \mrow{2}{0} & \mrow{2}{$c( \lceil\frac{a}{s_0}\rceil \lceil \frac{b}{s_1}\rceil + \lceil \frac{a}{s_0} \rceil )$} & \mrow{2}{$c$} & \mrow{2}{0} & \mrow{2}{$c \lceil\frac{a}{s_0}\rceil \lceil \frac{b}{s_1}\rceil + c$} & $1 + \frac{c}{2}(\lceil \frac{a}{s_0} \rceil + \lceil \frac{b}{s_1}\rceil ) - \lceil \frac{a}{s_0}\rceil$   \\
\cline{7-7}
 &  &  &  &  &  &  $1 + \frac{c}{2}(\lceil \frac{a}{s_0} \rceil \lceil \frac{b}{s_1}\rceil + \lceil \frac{b}{s_1}\rceil ) - \lceil \frac{a}{s_0}\rceil$ \\
\midrule
$\mathsf{Mult}$ & $bc$  & $c \lceil\frac{a}{s_0}\rceil \lceil \frac{b}{s_1}\rceil$ & $\frac{c}{2} \lceil \frac{a}{s_0} \rceil \lceil \frac{b}{s_1} \rceil $ & $bc$ & $c \lceil\frac{a}{s_0}\rceil \lceil \frac{b}{s_1}\rceil$ & $\frac{c}{2} \lceil \frac{a}{s_0} \rceil \lceil \frac{b}{s_1} \rceil$ \\
\midrule
\mrow{2}{$\mathsf{Rot}$} & \mrow{2}{0} & \mrow{2}{$c(\lceil \frac{b}{s_1} \rceil \log s_0 + \lceil \frac{a}{s_0} \rceil (\log s_1 + 1)) - \lceil \frac{a}{s_0} \rceil $} & \mrow{2}{$\frac{c}{2}(1 + 2 \lceil \frac{a}{s_0}\rceil \log s_1)$} & \mrow{2}{$bc \log s$}  & \mrow{2}{$c(\lceil \frac{b}{s_1} \rceil \log s_0 + \lceil \frac{a}{s_0} \rceil (\log s_1+ 1)) - \lceil \frac{a}{s_0} \rceil $} & $2 + \frac{c}{2} (2\lceil \frac{a}{s_0} \rceil  + \log s_0 \lceil \frac{b}{s_1} \rceil) - 2\lceil \frac{a}{s_0} \rceil$ \\
\cline{7-7}
 &  &  &  &  & &  $2 + \frac{c}{2} (\lceil \frac{a}{s_0} \rceil + \lceil \frac{a}{s_0}\rceil \lceil \frac{b}{s_1} \rceil + \log s_0 \lceil \frac{b}{s_1} \rceil) - \lceil \frac{a}{s_0} \rceil - \lceil \frac{a}{s_0} \rceil \lceil \frac{b}{s_1} \rceil$\\
\bottomrule
\end{tabular}
}
\caption{Exact number of operations in terms of $a, b, c, s_0, s_1$. Here $s = s_0 s_1$. For \textsf{DiagATB}, the first row (resp. second row) of \textsf{CMult} and \textsf{Rot} represent the case when $\mathsf{level}(A) \geq \mathsf{level}(B)$ (resp. $\mathsf{level}(A) < \mathsf{level}(B)$).}

\label{tab:matmul_cnt_formula}
\end{table*}

\end{document}